%% file: main.tex
\documentclass[a4paper,USenglish,cleveref, autoref, thm-restate]{lipics-v2021}
\pdfoutput=1

\usepackage{microtype}
\usepackage[utf8]{inputenc}
\usepackage{tikz}
\usepackage{mathtools}
\usepackage{hyperref}
\usepackage[ruled, vlined]{algorithm2e}
\usepackage{multirow}
\hypersetup{colorlinks=true,citecolor=blue}
\usetikzlibrary{automata}

\input{macros}

\title{Timed Games with Bounded Window Parity Objectives}

\author{James C.~A.~Main}{F.R.S.-FNRS \& UMONS -- Université de Mons, Belgium}{}{}{F.R.S.-FNRS Research Fellow.}

\author{Mickael Randour}{F.R.S.-FNRS \& UMONS -- Université de Mons, Belgium}{}{}{F.R.S.-FNRS Research Associate, member of the TRAIL Institute.}

\author{Jeremy Sproston}{Universit\`{a} degli Studi di Torino, Italy}{}{}{}

\authorrunning{J.~C.~A.~Main, M.~Randour and J.~Sproston} 

\Copyright{James C.~A.~Main, Mickael Randour, and Jeremy Sproston} 



\category{} 

\relatedversion{} 

\supplement{}

\funding{}

\nolinenumbers 

\hideLIPIcs  

\EventEditors{John Q. Open and Joan R. Access}
\EventNoEds{2}
\EventLongTitle{42nd Conference on Very Important Topics (CVIT 2016)}
\EventShortTitle{CVIT 2016}
\EventAcronym{CVIT}
\EventYear{2016}
\EventDate{December 24--27, 2016}
\EventLocation{Little Whinging, United Kingdom}	
\EventLogo{}
\SeriesVolume{42}
\ArticleNo{23}

\begin{document}

\maketitle

\begin{abstract}
  The window mechanism, introduced by Chatterjee et al.~\cite{Chatterjee0RR15}
  for mean-payoff and total-payoff objectives in two-player turn-based games
  on graphs, refines long-term objectives with time bounds. This mechanism
  has proven useful in a variety of
  settings~\cite{DBLP:journals/corr/BruyereHR16,DBLP:journals/lmcs/BrihayeDOR20}, and most
  recently in timed systems~\cite{MRS21}.

  In the timed setting, the so-called fixed timed window parity objectives
  have been studied. A fixed timed window parity objective is defined with
  respect to some time bound
  and requires that, at all times, we witness a time frame, i.e., a
  window, of size less than the fixed bound in which the smallest priority
  is even. In this work, we focus on the bounded timed window parity objective.
  Such an objective is satisfied
  if there exists some bound for which the fixed objective is satisfied.
  The satisfaction of bounded objectives is robust to modeling choices
  such as constants appearing in constraints, unlike fixed objectives, for which
  the choice of constants may affect the satisfaction for a given bound.

  We show that verification of bounded timed window objectives in timed automata
  can be performed in polynomial space, and that timed games with these
  objectives can be solved in exponential time, even for multi-objective
  extensions. This matches the complexity classes of the fixed case.
  We also provide a comparison of the different variants of window parity
  objectives.
\end{abstract}

\section{Introduction}\label{section:introduction}
\subparagraph*{Real-time systems.} \textit{Timed automata}~\cite{AlurD94} 
are a means of modeling systems in which the passage of time is
critical.
A timed automaton is a finite automaton extended with a set of
\textit{real-valued} variables called \textit{clocks}. All clocks of a timed
automaton increase at the same rate and measure the elapse of time in
a continuous fashion. Clocks constrain transitions in
timed automata and can be reset on these transitions.

Timed automata provide a formal setting for the verification of real-time
systems~\cite{AlurD94,DBLP:books/daglib/0020348}. When analyzing timed automata,
we usually exclude some unrealistic behaviors.
More precisely, we ignore
\textit{time-convergent paths}, i.e., infinite paths in which the total
elapsed time is bounded.
Even though timed automata induce uncountable transition systems, many
properties can be checked using the \textit{region abstraction}, a finite
quotient of the transition system.

Timed automata can also be used to design correct-by-construction
controllers for real-time systems. To this end, we model the interaction of
the system and its uncontrollable environment as a timed automaton
game~\cite{MalerPS95},
or more simply a \textit{timed game}. A timed game is a two-player game played
on a timed automaton by the system and its environment for an infinite number
of rounds. At each round, both players propose a real-valued delay and an
action, and the play progresses following the fastest move.

The notion of winning in a timed game must take time-convergence in account;
following~\cite{AlfaroFHMS03}, we declare as winning the plays that are either
\textit{time-divergent and satisfy the objective of the player}, or that
are \textit{time-convergent and the player is not responsible for convergence}.

\subparagraph*{Parity conditions.} Parity conditions are a canonical way of
specifying $\omega$-regular conditions, such as safety and liveness. A parity
objective is defined from a priority function, which assigns a
non-negative integer to each location of a timed automaton. The
\textit{parity objective} requires that the smallest priority witnessed
infinitely often is even.

\subparagraph*{The window mechanism.} A parity objective requires that
for all odd priorities seen infinitely often, there is some smaller even
priority seen infinitely often. However, the parity objective does not enforce
timing constraints; the parity objective can be satisfied despite there
being arbitrarily large delays between odd priorities and smaller even
priorities. Such behaviors may be undesirable, e.g., if odd priorities model
requests in a system and even priorities model responses.

The window mechanism was introduced by Chatterjee et al.~for mean-payoff games
in graphs~\cite{Chatterjee0RR15} and later applied to parity games in
graphs~\cite{DBLP:journals/corr/BruyereHR16} and mean-payoff and parity
objectives in Markov decision processes~\cite{DBLP:journals/lmcs/BrihayeDOR20}.
It is a means of reinforcing the parity objective
with \textit{timing constraints}.
A \textit{direct fixed timed window parity objective}
for some fixed time bound requires that at all times, we witness a
\textit{good window}, i.e., a time frame of size less than the fixed bound
in which the smallest priority is even. In other words, this objective requires
that the parity objective be locally satisfied at all times, where the notion
of locality is fixed in the definition.
This window parity objective and a prefix-independent variant
requiring good windows from some point on were studied in~\cite{MRS21}.

The main focus of this article is another variant of timed window parity
objectives called \textit{direct bounded timed window parity objectives}, which
extend the bounded window parity objectives of~\cite{DBLP:journals/corr/BruyereHR16}. This
objective is satisfied if and only if there exists some time bound for which
the direct fixed objective is satisfied. While this objective also requires that
the parity objective be locally satisfied at all times, the notion of
locality is \textit{not fixed a priori}. In particular, unlike the fixed
objective, its satisfaction is robust to modeling choices such as the choice
of constants constraining transitions, and depends only on the high-level
behavior of
the system being modeled. In addition to this direct objective, we also
consider a prefix-independent variant, the
\textit{bounded timed window parity objective}, which requires that some suffix
satisfies a direct bounded objective.

\subparagraph*{Contributions.} We study conjunctions of (respectively direct)
bounded timed window parity objectives in the setting of timed
automata and of timed games. We show that
checking that all time-divergent paths of a timed automaton satisfy a
conjunction of (respectively direct) bounded timed window parity objectives
can be done in \textsf{PSPACE} (Theorem~\ref{theorem:verification:complexity}).
We also show
that if all time-divergent paths of a timed automaton satisfy a (respectively
direct) bounded timed window parity objective, then there exists a bound
for which the corresponding fixed objective is satisfied
(Corollaries~\ref{corollary:direct:uniformity}
and~\ref{corollary:indirect:uniformity}).

In timed games, we show that in the direct case, the set of winning states can
be computed in \textsf{EXPTIME} (Theorem~\ref{theorem:games:direct:complexity})
by means of a timed game with an $\omega$-regular request-response
objective~\cite{DBLP:conf/wia/WallmeierHT03,DBLP:conf/lata/ChatterjeeHH11}. We show that, assuming a global
clock that cannot be reset, finite-memory strategies suffice to win, and
if a winning strategy for a direct bounded objective exists, there exists
a finite-memory winning strategy that is also winning for a direct fixed
objective (Theorem~\ref{theorem:games:direct:winning}).
In the prefix-independent case, we provide a fixed-point algorithm
to compute the set of winning states that runs in \textsf{EXPTIME}
(Theorem~\ref{theorem:games:indirect:complexity}). We
infer from the correctness proof that, assuming a global clock, finite-memory
strategies suffice for winning and if a winning strategy exists, then there
exists a finite-memory winning strategy that is also winning for some fixed
objective (Theorem~\ref{theorem:games:indirect:winning}).

We complement all membership results above with lower bounds and establish
\textsf{PSPACE}-completeness for timed automata-related problems and
\textsf{EXPTIME}-completeness for timed games-related problems
(Theorem~\ref{theorem:complexity:completeness}). 

\subparagraph*{Comparison.} Window objectives strengthen classical objectives
with timing constraints; they provide
\textit{conservative approximations} of these objectives
(e.g.,~\cite{Chatterjee0RR15,DBLP:journals/corr/BruyereHR16,DBLP:journals/lmcs/BrihayeDOR20}).
The complexity of window objectives, comparatively to that of the related
classical objective, depends on whether one considers a single-objective or
multi-objective setting.
In turn-based games on graphs, window objectives provide
\textit{polynomial-time} alternatives to the classical
objectives~\cite{Chatterjee0RR15,DBLP:journals/corr/BruyereHR16} in the
single-objective setting, despite, e.g., turn-based parity games on graphs
not being known to be solvable in polynomial time (parity games were recently
shown
to be
solvable in quasi-polynomial time~\cite{CaludeJKL017}). On the other hand, in the
multi-objective setting, the complexity is higher than that of the classical
objectives; for instance, solving a turn-based game with a conjunction of
fixed (respectively bounded) window parity objectives is
\textsf{EXPTIME}-complete~\cite{DBLP:journals/corr/BruyereHR16}, whereas solving
games with conjunctions of parity objectives is \textsf{co-NP}
complete~\cite{DBLP:conf/fossacs/ChatterjeeHP07}. In the timed setting, we
establish that solving timed games with conjunctions of bounded timed window
parity objectives is \textsf{EXPTIME}-complete, i.e., dense time comes for
free, similarly to the fixed case in timed games~\cite{MRS21}.

Timed games with classical parity objectives can be solved in exponential
time~\cite{AlfaroFHMS03,ChatterjeeHP11}, i.e., the complexity class of solving
timed games with window parity objectives matches that of solving timed games
with classical parity objectives. Timed games with parity objectives can
be solved by means of a reduction to an untimed parity game played on a graph
polynomial in the size of the region abstraction and the number of
priorities~\cite{ChatterjeeHP11}. However, most algorithms for games on graphs
with parity objectives suffer from a blow-up in complexity due to the number of
priorities. Timed window parity objectives provide an alternative to parity
objectives that bypasses this blow-up; in particular, we show in this paper that
\textit{timed games with a single bounded timed window objective can be solved
  in time polynomial in the size of the region abstraction and the number of
  priorities}.

In timed games, we show that winning for a (respectively direct) bounded timed
window parity objective is equivalent to winning for a (respectively direct)
fixed timed window parity objective with some sufficiently large bound that
depends on the number of priorities, number of objectives and the size of the
region abstraction. Despite the fact that this bound can be directly
computed (Theorems~\ref{theorem:games:direct:winning}
and~\ref{theorem:games:indirect:winning}), solving timed games with
(respectively direct) fixed timed window parity objectives for a certain bound
takes time that is polynomial in the size of the region abstraction, the
number of priorities and the \textit{fixed bound}. This bound may be large;
the algorithms we provide
for timed games with (respectively direct) bounded timed window parity
objectives avoid this additional contribution to the complexity.

\subparagraph*{Related work.} The window mechanism has seen numerous extensions
in addition to the previously mentioned works, e.g.,~\cite{DBLP:conf/csl/BaierKKW14,DBLP:conf/rp/Baier15,DBLP:conf/concur/BrazdilFKN16,DBLP:conf/concur/BruyereHR16,DBLP:journals/acta/HunterPR18,DBLP:conf/fsttcs/0001PR18,DBLP:conf/fsttcs/BordaisGR19}. Window parity objectives, especially bounded variants, are
closely related to the notion of finitary $\omega$-regular games,
e.g.,~\cite{DBLP:journals/tocl/ChatterjeeHH09}, and the semantics of
\textsc{prompt-ltl}~\cite{DBLP:journals/fmsd/KupfermanPV09}.
The window mechanism can be used to ensure a certain form of (local) guarantee over paths; different techniques have been considered in stochastic models~\cite{DBLP:journals/jcss/BrazdilCFK17,DBLP:journals/iandc/BruyereFRR17,DBLP:conf/icalp/BerthonRR17}.
Timed automata have numerous extensions, e.g., hybrid systems
(e.g.,~\cite{DBLP:journals/corr/abs-2001-04347} and references therein) and
probabilistic timed automata (e.g.,~\cite{DBLP:journals/fmsd/NormanPS13});
the window mechanism could prove useful in these richer settings.
Finally, we recall that game models provide a framework for the synthesis of
correct-by-construction controllers~\cite{rECCS}.

\subparagraph*{Outline.} Section~\ref{section:prelim} presents all preliminary
notions. Window objectives, relations between them and a useful property
of bounded window objectives are presented in
Section~\ref{section:objectives}. The verification of bounded window objectives
in timed automata is studied in Section~\ref{section:verification}.
Section~\ref{section:games} presents algorithms for timed games with bounded
window objectives. Lower bounds for completeness of the verification and
realizability problems for bounded window objectives are provided in
Section~\ref{section:completeness}. Finally, in
Section~\ref{section:comparison}, we compare the untimed and timed settings,
and the fixed and bounded objectives.
Appendix~\ref{appendix:strategies} expands
upon the preliminaries and discusses winning strategies in timed games with
$\omega$-regular objectives.

\section{Preliminaries}\label{section:prelim}

\subparagraph*{Notation.} We denote the set of non-negative real numbers by
$\IR_{\geq 0}$, and the set of non-negative integers by $\IN$.
Given some non-negative real number $x$, we write
$\lfloor x\rfloor$ for the integral part of $x$ and
$\fracpart(x) = x - \lfloor x\rfloor$ for its fractional part.
Given two sets $A$ and $B$, we let $2^A$ denote the power set of $A$ and
$A^B$ denote the section of functions $B\to A$.

\subparagraph*{Timed automata.}
A clock variable, or \textit{clock}, is a real-valued variable.
Let $C$ be a set of clocks. A \emph{clock constraint} over
$C$ is a conjunction of formulae of the form $x \sim c$ with
$x\in C$, $c\in \IN$, and
$\sim\in \{\leq,\geq, >, <\}$.  We write $x=c$ as shorthand for the clock
constraint $x\geq c\land x\leq c$. Let $\Phi(C)$ denote the set
of clock constraints over $C$.

We refer to functions $v\in \IR_{\geq 0}^C$ as
\emph{clock valuations} over $C$. A clock valuation $v$ over a set
$C$ of clocks satisfies a clock constraint of the form  $x\sim c$ if
$v(x) \sim c$ and $v$ satisfies
a conjunction $g\land h$ of two clock constraints $g$ and $h$
if it satisfies both
$g$ and $h$. Given a clock constraint $g$ and clock valuation $v$, we
write $v\models g$ if $v$ satisfies $g$.

For a clock valuation $v$ and $\delay \geq 0$, we let
$v + \delay$ be the valuation defined by $(v + \delay)(x) = v(x) + \delay $ for
all $x\in C$. For any valuation $v$ and $D\subseteq C$,
we define $\mathsf{reset}_D(v)$ to be the valuation agreeing with
$v$ for clocks in $C\setminus D$ and that assigns $0$ to clocks in $D$.
We denote by $\mathbf{0}^C$ the zero valuation, assigning 0 to all clocks in
$C$.

A \emph{timed automaton} (TA) is a tuple $(L, \ell_\init, C, \Sigma, I, E)$ where
$L$ is a finite set of \emph{locations}, $\ell_\init\in L$ is an initial
location, $C$ a finite set of \emph{clocks} containing a special
clock $\globalclock$ which keeps track of the total time elapsed,
$\Sigma$ a finite set of actions,
$I\colon L\to \Phi(C)$ an \emph{invariant} assignment function and
$E\subseteq L\times \Phi(C)\times \Sigma\times 2^{C\setminus\{\globalclock\}}\times L$
a finite edge relation.
We only consider deterministic timed automata, i.e., we assume that
in any location $\ell$, there are no two different outgoing edges
$(\ell, g_1, a, D_1, \ell_1)$ and $(\ell, g_2, a, D_2, \ell_2)$
sharing the same action
such that the conjunction $g_1\land g_2$ is satisfiable. For an edge
$(\ell, g, a, D, \ell')$, the clock constraint $g$ is called the
\textit{guard} of the edge.

A TA $\automaton = (L, \ell_\init, C, \Sigma, I, E)$
gives rise to an uncountable transition system
$\mathcal{T}(\automaton) = (S, s_\init, M, \to)$
with the state space $S = L\times \IR_{\geq 0}^C$,
the initial state
$s_\init= (\ell_\init, \mathbf{0}^C)$, set of transition system
actions $M = \IR_{\geq 0}\times
(\Sigma\cup\{\bot\})$ and the
transition relation $\to\,\subseteq S\times M\times S$ defined as follows:
for any action $a\in\Sigma$ and
delay $\delay\geq 0$, we have that $((\ell, v), (\delay, a), (\ell', v'))\in\,\to$ if
and only if there is some edge $(\ell, g, a, D, \ell')\in E$
such that $v + \delay\models g$, $v' = \mathsf{reset}_{D}(v + \delay)$,
$v + \delay\models I(\ell)$ and $v'\models I(\ell')$; for any delay
$\delay\geq 0$, $((\ell, v)(\delay, \bot),(\ell, v + \delay))\in\,\to$
if $v + \delay\models I(\ell)$.
Let us note that the satisfaction set of clock constraints is convex: it
is described by a conjunction of inequalities. Whenever
$v\models I(\ell)$, the above condition $v + \delay\models I(\ell)$
(the invariant holds after the delay)
is equivalent to requiring $v + \delay'\models I(\ell)$ for all
$0\leq \delay'\leq \delay$
(the invariant holds at each intermediate time step).

A \textit{move} is any pair in $\IR_{\geq 0}\times (\Sigma\cup\{\bot\})$
(i.e., an action in the transition system).
For any move $m=(\delay, a)$ and states $s$, $s'\in S$, we write
$s\xrightarrow{m}s'$ or $s\xrightarrow{\delay, a}s'$ as shorthand for
$(s,m,s')\in\,\to$.
Moves of the form $(\delay, \bot)$ are called \textit{delay moves}.
For any move $m=(\delay, a)$, we let $\delayfunc(m) = \delta$.
We say a move $m$ is
enabled in a state $s$ if there is some $s'$ such that $s\xrightarrow{m}s'$.
There is at most one successor per move in a state, as we do not allow two
guards on edges labeled by the same action to be simultaneously satisfied.

A \textit{path} in a TA $\automaton$ is a finite or infinite
sequence
$s_0m_0s_1\ldots\in S(MS)^*\cup (SM)^\omega$ such that for all $j\in\IN$, $s_j$
is a state of $\mathcal{T}(\automaton)$ and
$s_{j}\xrightarrow{m_{j}}s_{j+1}$ is a
transition in $\mathcal{T}(\automaton)$. A path is \textit{initial}
if $s_0=s_{\init}$. For clarity, we write
$s_0\xrightarrow{m_0}s_1\xrightarrow{m_1}\cdots$ instead
of $s_0m_0s_1\ldots$.

A state $s$ is said to be \textit{reachable from a state} $s'$
if there exists a path from $s'$ to $s$. Similarly, a set of states
$T\subseteq S$ is said to be reachable from some state $s'$ if there is
a path from $s'$ to a state in $T$. We say that a state is \textit{reachable} if
it is reachable from the initial state.

An infinite path
$\pi = (\ell_0, v_0)\xrightarrow{m_0}(\ell_1, v_1)\xrightarrow{m_1}\ldots$
is \textit{time-divergent} if the sequence $(v_j(\globalclock))_{j\in\IN}$
is not bounded from above. A path that is not time-divergent is called
\textit{time-convergent}; time-convergent paths are traditionally
ignored in analysis
of timed automata \cite{DBLP:journals/iandc/AlurCD93,AlurD94} as they model unrealistic behavior.
This includes ignoring
\emph{Zeno paths}, which are time-convergent paths along which infinitely
many actions appear.

\subparagraph*{Regions.} The transition system induced by a TA is infinite.
Qualitative properties of TAs can nonetheless be analyzed using the
region abstraction~\cite{AlurD94}, a quotient of the transition
system by an equivalence relation of finite index.
Fix a TA $\automaton= (L, \ell_\init, C, \Sigma, I, E)$.
For each clock $x\in C$, let $c_x$ denote the largest constant to which
$x$ is compared to in guards and invariants of $\automaton$.

We define an equivalence relation over clock valuations of $C$: we say that
two clock valuations $v$ and $v'$ over $C$
are \textit{clock-equivalent} for $\automaton$, denoted by
$v\clockequiv v'$,
if the following properties are satisfied:
(i) for all clocks $x\in C$, $v(x) > c_x$ if and only if $v'(x) > c_x$;
(ii) for all clocks $x\in\{z\in C\mid v(z)\leq c_z\}$,
$\lfloor v(x) \rfloor = \lfloor v'(x) \rfloor$;
(iii) for all clocks $x, y\in\{z\in C\mid v(z)\leq c_z\}\cup\{\globalclock\}$,
$v(x)\in\IN$ if and only if $v'(x)\in\IN$, and
$\fracpart(v(x))\leq \fracpart(v(y))$ if and only if
$\fracpart(v'(x))\leq \fracpart(v'(y))$. When $\automaton$ is clear from the
context, we say that two valuations are clock-equivalent rather than
clock-equivalent for $\automaton$.

An equivalence class for this
relation is referred to as a \textit{clock region}. We denote the
equivalence class for $\clockequiv$ of a clock valuation $v$ as $[v]$.
We let $\regions$ denote the set of all clock regions.
The number of clock regions is finite, and exponential in the number of clocks
and the encoding of the constants $c_x$, $x\in C$. More precisely, we have the bound
$|\regions|\leq |C| ! \cdot 2^{|C|}\cdot \prod_{x\in C} (2c_x + 1)$.

We extend the equivalence defined above to states as well.
We say that two states $s= (\ell, v)$ and $s' = (\ell', v')$ are
\textit{state-equivalent}, denoted $s\clockequiv s'$, whenever
$\ell = \ell'$ and $v\clockequiv v'$. An equivalence class for this relation
is referred to as a \textit{state region}. Given some state $s\in S$, we write
$[s]$ for its equivalence class. We identify the set of state regions with
the set $L\times \regions$ and sometimes denote state regions
as pairs $(\ell, R)\in L\times\regions$ in the sequel.

The satisfaction of clock constraints that appear in $\automaton$ is uniform
inside of a clock region.
For a clock region $R \in \regions$ and a clock constraint $g$ of
$\mathcal{A}$, we write $R \models g$ to denote that
$v \models g$ for some $v \in R$.
This does not hold for arbitrary
clock constraints, e.g., it does not hold for clock constraints involving
constants larger than those in $\automaton$.
The reset operator also respects regions, in the sense that for
any clock valuation $v$ and $D\subseteq C$, $[\reset_D(v)] = \{\reset_D(v')\mid v'\in [v]\}$.
Let $R$, $R'$ be two clock regions. We say that $R'$ is a
\textit{successor region} of $R$ if for all valuations $v\in R$, there exists
some delay $\delay_v\geq 0$ such that $v + \delay_v\in R'$.

We now have all of the notions required to define the region abstraction
of $\mathcal{T}(\automaton)$. The \textit{region abstraction} of
$\mathcal{T}(\automaton)$ is the finite transition system
$(L \times \regions, [s_\init], \{\tau\}, \to')$ where $L\times \regions$ is
the state space, elements of which are state regions, with the state region
of the initial state as its initial state, a unique move $\tau$
(we abstract the actions of the TA away), and a transition relation $\to'$
defined as follows. For any two state regions
$(\ell, R)$, $(\ell', R')\in L\times\regions$,
$(\ell, R)\xrightarrow{\tau}' (\ell', R')$ holds if and only if one of the
two following conditions hold: (i) there exists some edge
$(\ell, g, a, D, \ell')\in E$ and some successor region $R_{\textsf{succ}}$ of
$R$ such that $R_{\textsf{succ}}\models g\land I(\ell)$,
$R' = \{\reset_D(v) \mid v\in R_{\textsf{succ}}\}$ and $R'\models I(\ell')$, or
(ii) $\ell = \ell'$ and $R'$ is a successor region of $R$ such that
$R'\models I(\ell)$. These conditions respectively abstract transitions that
use edges and delay transitions in $\mathcal{T}(\automaton)$.

Any (finite or infinite) path $s_0\xrightarrow{m_0} s_1\xrightarrow{m_1}\ldots$
of $\mathcal{T}(\automaton)$ induces a path
$[s_0]\xrightarrow{\tau}'[s_1]\xrightarrow{\tau}'\ldots$ in the region
abstraction. Conversely, for any path
$(\ell_0, R_0)\xrightarrow{\tau}'(\ell_1, R_1)\xrightarrow{\tau}'\ldots$
and any $v_0\in R_0$, one can find a path of
$\mathcal{T}(\automaton)$
$(\ell_0, v_0)\xrightarrow{m_0}(\ell_1, v_1)\xrightarrow{m_1}\ldots$
such that $v_n\in R_n$ for all $n\in\IN$. Due to this relation between paths,
qualitative properties that depend only on locations or regions can be
verified using the region abstraction.

\subparagraph*{Priorities.} A \textit{priority function} is a function
$p\colon L\to \{0, \ldots, \maxpriority - 1\}$ with $\maxpriority\leq |L| + 1$.
We use priority functions
to express parity objectives. A
\textit{$\numdimensions$-dimensional priority function}
is a function $p\colon L\to \{0, \ldots, \maxpriority-1\}^\numdimensions$
which assigns vectors of priorities to locations.
Given a $\numdimensions$-dimensional priority function $p$ and
a dimension $k\in \{1, \ldots, \numdimensions\}$, we write $p_k$ for the
priority function given by $p$ on dimension $k$.

\subparagraph*{Timed games.}
We consider two player games played on TAs. We refer to the players
as player 1 ($\player_1$) for the system and player 2 ($\player_2$) for
the environment.
We use the notion of timed automaton games of \cite{AlfaroFHMS03}.

A \textit{timed} (automaton) \textit{game} (TG) is a tuple
$\game = (\automaton,\Sigma_1, \Sigma_2)$ where
$\automaton= (L, \ell_\init, C, \Sigma, I, E)$ is a TA and
$(\Sigma_1, \Sigma_2)$ is a partition of $\Sigma$.
We refer to actions in $\Sigma_i$ as $\player_i$ actions for $i\in\{1, 2\}$.

Recall that a move is a pair $(\delay, a)\in\IR_{\geq 0}\times (\Sigma\cup\{\bot\})$.
Let $S$ denote the set of states of $\mathcal{T}(\automaton)$.
In each state $s= (\ell, v)\in S$, the moves available
to $\player_1$ are the elements of the set $M_1(s)$ where
\[M_1(s) = \big\{ (\delay,a) \in \mathbb{R}_{\geq 0} \times (\Sigma_1 \cup \{ \bot \}) \mid \exists s', s \xrightarrow{\delay,a} s'  \big\}\]
contains moves with $\player_1$ actions and delay moves that are
enabled in $s$. The set $M_2(s)$ is defined
analogously with $\player_2$ actions. We write $M_1$ and $M_2$ for
the set of all moves of $\player_1$ and $\player_2$ respectively.

At each state $s$ along a play, both players simultaneously
select a move $m^{(1)}\in M_1(s)$ and $m^{(2)}\in M_2(s)$.
Intuitively, the fastest player gets to act and in case of
a tie, the move is chosen non-deterministically. This is
formalized by the
\emph{joint destination function} $\jointdestination: S\times M_1\times M_2\to 2^S$,
defined by
\[\jointdestination(s, m^{(1)}, m^{(2)}) = \begin{cases}
    \{s' \in S\mid s \xrightarrow{m^{(1)}} s'\} &
      \text{if } \delayfunc(m^{(1)}) < \delayfunc(m^{(2)}) \\
    \{s' \in S\mid s \xrightarrow{m^{(2)}} s'\} & \text{if } \delayfunc(m^{(1)}) > \delayfunc(m^{(2)}) \\
    \{s' \in S\mid s \xrightarrow{m^{(i)}} s', i=1, 2\} & \text{if } \delayfunc(m^{(1)}) = \delayfunc(m^{(2)}).
  \end{cases}\]
For $m^{(1)}=(\delay^{(1)}, a^{(1)})\in M_1$ and
$m^{(2)}=(\delay^{(2)}, a^{(2)})\in M_2$, we write
$\delayfunc(m^{(1)}, m^{(2)})=\min\{\delay^{(1)}, \delay^{(2)}\}$ to denote the delay occurring when
$\player_1$ and $\player_2$ play $m^{(1)}$ and $m^{(2)}$ respectively.

A play is defined similarly to an infinite path:  a \textit{play} is
an infinite sequence of the form
$s_0(m_0^{(1)}, m_0^{(2)})s_1(m_1^{(1)}, m_1^{(2)})\ldots\in
(S(M_1\times M_2))^\omega $ where
for all indices $j\in\IN$, $m_j^{(i)}\in M_i(s_j)$ for $i\in\{1, 2\}$ and
$s_{j+1}\in \jointdestination(s_{j+1}, m_{j+1}^{(1)}, m_{j+1}^{(2)})$. A
\textit{history} is a finite prefix of a play ending in a state.
A play or history $s_0(m_0^{(1)}, m_0^{(2)})s_1\ldots$  is \textit{initial} if
$s_0=s_{\init}$. For any history
$h = s_0(m_0^{(1)}, m_0^{(2)})\ldots (m_{n-1}^{(1)}, m_{n-1}^{(2)})s_n$, we
set $\last(h) = s_n$. For a play $\pi = s_0(m_0^{(1)}, m_0^{(2)})s_1\ldots$,
we write $\pi_{|n} = s_0(m_0^{(1)}, m_0^{(2)})\ldots (m_{n-1}^{(1)}, m_{n-1}^{(2)})s_n$.
Plays of $\game$ follow paths of $\automaton$. For a play, there may be several such paths:
if at some point of the play both players use a move with the same delay and
successor state, either move can label the transition in a matching path.

Similarly to paths, a play
$\pi=(\ell_0, v_0)(m_0^{(1)}, m_0^{(2)})\cdots$
is \textit{time-divergent} if and only if $(v_j(\globalclock))_{j\in\IN}$
is not bounded from above. Otherwise, we say a play is \textit{time-convergent}.
We define the following sets:
$\plays(\game)$ for the set of plays of $\game$;
$\hist(\game)$ for the set of histories of $\game$;
$\plays_\infty(\game)$ for the set of time-divergent plays of $\game$.
We also write $\plays(\game, s)$ to denote plays starting in state $s$ of
$\mathcal{T}(\automaton)$.

We will deal with objectives that require properties in the limit.
For this purpose, we introduce a notation for suffixes of plays.
For any $n\in\IN$ and any play
$\pi = s_0(m_0^{(1)}, m_0^{(2)})s_1\ldots\in \plays(\game)$,
we let $\pi_{n\to} = s_n(m_n^{(1)}, m_n^{(2)})s_{n+1} \ldots$ denote
the suffix of $\pi$ starting at index $n$.

\subparagraph*{Strategies.} A \textit{strategy} for $\player_i$ is a function describing
which move a player should use based on a history.
Formally, a strategy for $\player_i$ is a function
$\sigma_i\colon \hist(\game)\to M_i$ such that for all
$\pi\in\hist(\game)$, $\sigma_i(\pi)\in M_i(\last(\pi))$. This last
condition requires that each move given by a strategy
be enabled in the last state of a play.

A play or history $s_0(m_0^{(1)}, m_0^{(2)})s_1\ldots$ is said to be consistent
with a $\player_i$-strategy $\sigma_i$ if for all indices $j$,
$m_j^{(i)} = \sigma_i(\pi_{|j})$. Given a $\player_i$ strategy $\sigma_i$,
we define $\outcome_i(\sigma_i)$ (resp.~$\outcome_i(\sigma_i, s)$)
to be the set of plays (resp.~set of plays starting in state $s$)
consistent with $\sigma_i$.

In general, strategies can exploit full knowledge of the past, and need not
admit some finite representation. In the sequel, we focus on a subclass of
finite-memory strategies. A strategy is a finite-memory strategy if it can be
encoded by a finite Mealy machine, i.e., a deterministic automaton with outputs.
A \textit{Mealy machine} (for a strategy of $\player_1$) is a tuple
$\fmmealy  = (\fmstates, \mathfrak{m}_\init, \update, \nextmove)$
where $\fmstates$ is a finite set of states, $\mathfrak{m}_\init\in\fmstates$
is an initial state, $\update\colon\fmstates\times S\to\fmstates$ is the memory
update function and $\nextmove\colon\fmstates\times S\to M_1$ is the
next-move function.

Let $\fmmealy  = (\fmstates, \mathfrak{m}_\init, \update, \nextmove)$ be a Mealy
machine. We
define the strategy induced by $\fmmealy$ as follows. Let $\varepsilon$
denote the empty word. We first define the iterated update function
$\update^*\colon S^*\to \fmstates$ inductively as
$\update^*(\varepsilon) = \mathfrak{m}_\init$ and for any
$s_0\ldots s_n\in S^*$, we let
$\update^*(s_0\ldots s_n) = \update(\update^*(s_0\ldots s_{n-1}), s_n)$.
The strategy $\sigma$ induced by $\fmmealy$ is defined by
$\sigma(h) = \nextmove(\update^*(s_0\ldots s_{n-1}), s_n)$ for any
history
$h = s_0(m_0^{(1)}, m_0^{(2)})\ldots (m_{n-1}^{(1)}, m_{n-1}^{(2)})s_n\in
\hist(\game)$.
A strategy $\sigma$ is said to be \textit{finite-memory} if it is induced by
some Mealy machine.

We will exploit a subclass of finite-memory strategies that are well-behaved
with respect to regions. We say that some
strategy $\sigma$ is a finite-memory region strategy if it can be encoded
by a Mealy machine the updates of which depend only on the current region
(rather than the state itself) and such that, in a given memory state, the
moves proposed in two state-equivalent game states traverse the same state
regions during the proposed delay and then move to the same region.
Formally, a strategy is a \textit{finite-memory region strategy} it is
induced by some Mealy machine 
$\fmmealy  = (\fmstates, \mathfrak{m}_\init, \update, \nextmove)$
where for any memory states $\mathfrak{m}\in\fmstates$ and any two
state-equivalent states $s= (\ell, v)$, $s'=(\ell', v')\in S$,
$\update(\mathfrak{m}, s) = \update(\mathfrak{m}, s')$ and the moves
$(\delay, a) = \nextmove(\mathfrak{m}, s)$ and
$(\delay', a') = \nextmove(\mathfrak{m}, s')$ are such that $a = a'$,
$[v + \delay] = [v' + \delay']$ and
$\{[v + \delay_{\mathsf{mid}}]\mid 0\leq \delay_{\mathsf{mid}}\leq \delay\} =
\{[v' + \delay_{\mathsf{mid}}]\mid 0\leq \delay_{\mathsf{mid}}\leq \delay'\}$.
We view the update function of Mealy machines
inducing finite-memory region  strategies as functions
$\update\colon \fmstates\times L\times\regions\to\fmstates$.

\subparagraph*{Objectives.} An objective represents the
property we desire on paths of a TA or a goal of a player in a TG.
Formally, we define an \textit{objective} as a set
$\Psi\subseteq S^\omega$ of infinite sequences of states.
An objective $\Psi$ is a \textit{region objective} if given two
sequences of states $s_0s_1\ldots$, $s_0's_1'\ldots\in S^\omega$ such that
for all $j\in\IN$, $s_j\clockequiv s_j'$, we have $s_0s_1\ldots\in\Psi$
if and only if $s_0's_1'\ldots\in\Psi$. Intuitively, the satisfaction of
a region objective depends only on the witnessed sequence of state regions.

An $\omega$-regular region objective is a region objective
recognized by some deterministic parity automaton.
A (total) \textit{deterministic parity automaton} (DPA) is a tuple
$H = (\dpastates, q_\init, A, \dpatransitions, p)$, where $\dpastates$ is
a finite set of states, $q_\init\in \dpastates$ is the initial state, $A$
is a finite alphabet, 
$\dpatransitions\colon \dpastates\times A\to \dpastates$ is a
total transition function and $p: Q\to \{0, \ldots, \maxpriority-1\}$ is a priority
function over the states of the DPA.

Let $w= a_0a_1\ldots\in A^\omega$ be an infinite word. The execution of $H$
over $w$ is the infinite sequence of states $q_0q_1\ldots\in Q^\omega$ that
starts in the initial state of $H$, i.e., $q_0=q_\init$ and such that for all
$n\in\IN$, $q_{n+1}=\dpatransitions(q_n, a_n)$, i.e., each step of
the execution is performed by reading a letter of the input word.
An infinite word $w\in A^\omega$ is accepted by $H$ if the smallest priority
appearing infinitely often along the execution $q_0q_1\ldots\in Q^\omega$ over
$w$ is even, i.e. if $(\liminf_{n\to\infty}p(q_i))\bmod 2=0$. We denote by
$\mathcal{L}(H)$ the set of words accepted by $H$.

We use DPAs to encode $\omega$-regular objectives over state regions.
A DPA $H = (\dpastates, q_\init, L\times \regions, \dpatransitions, p)$
formally encodes the objective
$\{s_0s_1\ldots\in S^\omega\mid [s_0][s_1]\ldots\in\mathcal{L}(H)\}$.

In the sequel, we use the following $\omega$-regular region objectives in
addition to the window objectives studied in this work.
The window objectives we consider later on are derived from the
parity objective.
The \textit{parity} objective for a one-dimensional priority function
$p\colon L\to\{0, \ldots, \maxpriority-1\}$ requires that
the smallest priority seen infinitely often is even.
Formally, we define $\mathsf{Parity}(p) = \{
(\ell_0, v_0) (\ell_1, v_1)\ldots\in S^\omega
\mid (\liminf_{n\to\infty}p(\ell_n))\bmod 2 =0\}$.
For hardness arguments, we rely on safety
objectives. A \textit{safety} objective, defined with respect to a
set of locations $F\subseteq L$, requires
that no location in $F$ be visited. Formally, the safety objective for $F$
is defined as $\mathsf{Safe}(F) =\{
(\ell_0, v_0) (\ell_1, v_1)\ldots\in S^\omega \mid
\forall\, n,\, \ell_n\notin F\}$.

For the sake of brevity, given some path $\pi = s_0m_0s_1\ldots$ of a TA
or a play of a TG $\pi =s_0(m_0^{(1)}, m_0^{(2)})s_1\ldots$
and an objective $\Psi\subseteq S^\omega$ , we write $\pi\in \Psi$ to mean
that the sequence of states $s_0s_1\ldots$ underlying $\pi$ is in $\Psi$, and
say that $\pi$ satisfies the objective $\Psi$. 

\subparagraph*{Winning conditions.}
In games, we distinguish objectives and \emph{winning conditions}.
We adopt the definition of \cite{AlfaroFHMS03}. Let $\Psi$ be an objective.
It is desirable to have victory be achieved in a physically meaningful way:
for example, it is unrealistic to have a safety objective be achieved by
stopping time. This motivates a restriction to time-divergent plays.
However, this requires $\player_1$ to force the divergence of plays, which
is not reasonable, as $\player_2$ can stall using delays with zero time
units. Thus we also declare winning time-convergent plays where
$\player_1$ is \emph{blameless}.
Let $\textsf{Blameless}_1$ denote the set of $\player_1$-blameless
plays, which we define in the following way.

Let $\pi = s_0(m_0^{(1)}, m_0^{(2)})s_1\ldots$ be a play or a history. We say
$\player_1$ is \textit{not responsible} (or not to be blamed) for
the transition at step $k$ in $\pi$ if either $\delay_k^{(2)} < \delay_k^{(1)}$
($\player_2$ is faster) or $\delay_k^{(1)} = \delay_k^{(2)}$ and
$s_k\xrightarrow{\delay_k^{(1)}, a_k^{(1)}}s_{k+1}$
does not hold in $\mathcal{T}(\automaton)$
($\player_2$'s move was selected and did not have the same
target state as $\player_1$'s) where $m_k^{(i)} = (\delay_k^{(i)}, a_k^{(i)})$ for
$i\in\{1, 2\}$. 
The set $\textsf{Blameless}_1$ is formally defined as the set of infinite
plays $\pi$ such that there is some $j$ such that
for all $k\geq j$, $\player_1$ is not responsible for the transition at step
$k$ in $\pi$.

Given an objective
$\Psi$, we set the winning condition $\wc_1(\Psi)$ for $\player_1$ to be
the set of plays
\[\wc_1(\Psi) = (\plays(\game, \Psi)\cap\plays_\infty(\game))
  \cup (\textsf{Blameless}_1\setminus \plays_\infty(\game)),\]
where $\plays(\game, \Psi) = \{\pi\in\plays(\game)\mid\pi\in\Psi\}$.
Winning conditions for $\player_2$ are defined by exchanging the
roles of the players in the former definition.

We consider that the two players are adversaries and have opposite objectives,
$\Psi$ and $S^\omega\setminus \Psi$.
Let us note that
$\mathsf{WC}_1(\Psi)\cup \mathsf{WC}_2(S^\omega\setminus\Psi)\neq \plays(\game)$.
While this union subsumes all time-divergent plays and time-convergent plays
that are blameless for one player, it omits the time-convergent plays that
are blameless for neither player.

A \emph{winning strategy} for $\player_i$ for an objective $\Psi$ from
a state $s_0$ is a strategy $\sigma_i$ such that
$\outcome_i(\sigma_i, s_0)\subseteq \mathsf{WC}_i(\Psi)$.
We say that a state is winning for $\player_1$ for an objective $\Psi$ if
$\player_1$ has a winning strategy from this state.
\subparagraph*{Winning for $\omega$-regular region objectives.} Let us consider
a DPA $H = (\dpastates, q_\init, L\times \regions, \dpatransitions, p)$
with $p\colon \dpastates\to\{0, \ldots, \maxpriority - 1\}$
specifying an $\omega$-regular region objective in the TG $\game$. Let
$\maxpriority' = \maxpriority$ if $\maxpriority$ is odd and
$\maxpriority' = \maxpriority - 1$ otherwise.
The set of winning states for the objective $\mathcal{L}(H)$ is a union of state
regions and is computable in exponential
time~\cite{DBLP:conf/concur/AlfaroHM01,AlfaroFHMS03}.

\begin{restatable}{theorem}{theoremOmegaRegularComplexity}\label{theorem:omegaregular:complexity}
  The set of winning states of $\player_1$ in the TG $\game$ for
  the objective given by $H$ is a union of state regions and
  is computable in time
 $\bigo((4\cdot |L|\cdot|\regions|\cdot |Q|\cdot \maxpriority)^{\maxpriority'+2}).$
\end{restatable}

Furthermore, in TGs with $\omega$-regular objectives, finite-memory region
strategies suffice for winning. We can even obtain winning
finite-memory region strategies for which all delays
are bounded by some constant. Intuitively, if one replaces moves
of $\player_1$ of a winning strategy by delay moves with durations that are
bounded by some constant, one still has a winning strategy.
The broad justification is any outcome of the modified finite-memory region
strategy
shares its sequence of states with an outcome of the original strategy obtained
by having $\player_2$ interrupt the moves of $\player_1$ that have a large delay.
We therefore have the following, which is elaborated upon in
Appendix~\ref{appendix:strategies}.

\begin{restatable}{theorem}{theoremFMStrategies}\label{theorem:fm:strategies}
  There exists a finite-memory region strategy with
  $2\cdot |Q|\cdot \maxpriority$ states proposing delays of at most
  $1$ that is winning for the objective specified by $H$ from any state that is winning
  for $\player_1$.
\end{restatable}

\subparagraph*{Decision problems.}
We consider two different problems for an objective $\Psi$.
The first is the \emph{verification problem} for $\Psi$, which asks
given a TA whether all \emph{time-divergent initial} paths
satisfy the objective.
Second is the \emph{realizability problem} for $\Psi$, which
asks whether in a TG, $\player_1$
has a winning strategy from the initial state.

\section{Bounded window objectives}\label{section:objectives}

The main focus of this paper is a variant of timed window parity objectives
called \textit{bounded timed window parity objectives}. These are defined from
the \textit{fixed timed window parity objectives} studied in~\cite{MRS21}, where these
are referred to as timed window parity objectives. The definitions of the
different variants of timed window parity objectives are provided in
Section~\ref{section:objectives:definitions}. Section~\ref{section:objectives:relationships}
presents the relationships between the different variants and the original
parity objective.
Finally, 
Section~\ref{section:objectives:technical} introduces a technical result
used to simplify paths and plays witnessing the violation of a bounded window
objective.

For this entire section, we fix a TG $\game = (\automaton, \Sigma_1, \Sigma_2)$
where $\automaton = (L, \ell_\init, C, \Sigma_1\cup\Sigma_2, I, E)$ and
a one-dimensional priority function $p\colon L\to \{0, \ldots, \maxpriority-1\}$.

\subsection{Objective definitions}\label{section:objectives:definitions}

\subparagraph{Fixed objectives.} Fixed window objectives depend on a fixed time
bound $\lambda\in\IN$.
The first building block for the definition of window objectives is
the notion of good window.  A good window for the bound $\lambda$ is
intuitively a time interval of length strictly less than $\lambda$ for which the
smallest priority of the locations visited in the interval is even.
We define the
\textit{timed good window (parity) objective} as the set of sequences of states
that have a good window at their start. Formally, we define
$\mathsf{TGW}(p, \lambda)= \big\{(\ell_0, v_0)(\ell_1, v_1) \ldots
\in S^\omega\mid 
\exists \, n\in\IN,\, 
\min_{0\leq j\leq n}p(\ell_j)\bmod 2 = 0
\text{ and }  v_n(\globalclock)-v_0(\globalclock)<\lambda\big\}$.

The \textit{direct fixed timed window (parity) objective} for the bound
$\lambda$, denoted by $\mathsf{DFTW}(p, \lambda)$, requires
that the timed good window objective is satisfied by all suffixes of a sequence.
Formally, we define
$\mathsf{DFTW}(p, \lambda)$ as the set
$\{s_0s_1\ldots\in S^\omega\mid\forall\, n\in\IN,\,
s_ns_{n+1}\ldots\in\mathsf{TGW}(p, \lambda)\}.$

Unlike the parity objective, $\mathsf{DFTW}(p, \lambda)$ is not
prefix-independent. Therefore, a prefix-independent variant
of the direct fixed timed window objective, the
\textit{fixed timed window (parity) objective} $\mathsf{FTW}(p, \lambda)$, was
also studied in~\cite{MRS21}. Formally, we define
$\mathsf{FTW}(p, \lambda) = \{s_0s_1\ldots\in S^\omega\mid\exists\, n\in\IN,\,
s_ns_{n+1}\ldots\in\mathsf{DFTW}(p, \lambda)\}$.

\subparagraph{Bounded objectives.} A sequence of states satisfies the
(respectively direct) bounded timed window objective if there exists a time
bound $\lambda$ for which the sequence satisfies
the (respectively direct) fixed timed window objective.
Unlike the fixed case, this bound depends on the sequence of states, and
need not be uniform, e.g., among all sequences of states induced by
time-divergent
paths of a TA or among all sequences of states induced by time-divergent
outcomes of a strategy in a TG.

We formally define the
\textit{(respectively direct) bounded timed window (parity) objective}
$\mathsf{BTW}(p)$ (respectively $\mathsf{DBTW}(p)$) as the set
$\mathsf{BTW}(p)= \{s_0s_1\ldots\in S^\omega\mid
\exists\, \lambda\in\IN, s_0s_1\ldots\in\mathsf{FTW}(p, \lambda)\}$
(respectively $\mathsf{DBTW}(p)=\{s_0s_1\ldots\in S^\omega\mid
\exists\, \lambda\in\IN, s_0s_1\ldots\in\mathsf{DFTW}(p, \lambda)\}$).
The objective $\mathsf{BTW}(p)$ is a prefix-independent variant of
$\mathsf{DBTW}(p)$.

In the sequel, to distinguish the prefix-independent variants from
direct objectives, we may refer to the fixed timed window or bounded timed
window objectives as \textit{indirect objectives}.

\subparagraph{Multi-objective extensions.} In addition to the direct and indirect
bounded objectives, we will also study some of their multi-objective
extensions. More precisely, we assume for these definitions that $p$ is
a multi-dimensional priority function, i.e.,
$p\colon L\to \{0, \ldots, \maxpriority-1\}^\numdimensions$,
and define a multi-dimensional objective as the conjunction of the objectives
derived from the component functions $p_1$, \ldots, $p_\numdimensions$.

Multi-dimensional extensions are referred to as generalized objectives.
Formally, in the fixed case, for a bound $\lambda\in\IN$,
we define the \textit{generalized  direct fixed timed window objective} as
$\mathsf{GDFTW}(p, \lambda) = \bigcap_{1\leq k\leq \numdimensions}\mathsf{DFTW}(p_k, \lambda)$ and the \textit{generalized fixed timed window objective} as
$\mathsf{GFTW}(p, \lambda) = \bigcap_{1\leq k\leq \numdimensions}\mathsf{FTW}(p_k, \lambda)$.
In the bounded case, we define
the \textit{generalized direct bounded timed window objective} as
$\mathsf{GDBTW}(p) = \bigcap_{1\leq k\leq \numdimensions}\mathsf{DBTW}(p_k)$ and
the \textit{generalized bounded timed window objective} as
$\mathsf{GBTW}(p) = \bigcap_{1\leq k\leq \numdimensions}\mathsf{BTW}(p_k)$.

\subsection{Relationships between objectives}\label{section:objectives:relationships}

We discuss the relationships between the different timed window parity
objectives and the parity objective in this section. We discuss both inclusions and
differences between the different objectives.

The inclusions are induced by the fact that a direct objective is more
restrictive than its prefix-independent counterpart, and similarly, by the fact
that a fixed objective is more restrictive than its bounded counterpart. Parity
objectives, on the other hand, are less restrictive than any of the timed
window objectives, as they require no time-related aspect to hold.

\begin{lemma}\label{lemma:inclusions}
  The following inclusions hold for any $\lambda\in\IN$:
  \begin{itemize}
  \item $\mathsf{DFTW}(p, \lambda)\subseteq \mathsf{FTW}(p,\lambda)\subseteq \mathsf{BTW}(p) \subseteq \mathsf{Parity}(p)$ and
  \item $\mathsf{DFTW}(p, \lambda)\subseteq \mathsf{DBTW}(p)\subseteq \mathsf{BTW}(p) \subseteq \mathsf{Parity}(p)$.
  \end{itemize}
\end{lemma}
\begin{proof}
  We only argue that $\mathsf{BTW}(p) \subseteq \mathsf{Parity}(p)$, as all other
  inclusions are straightforward. Let $\pi = s_0s_1\ldots\in\mathsf{BTW}(p)$. It follows
  that $\pi$ has some suffix $\pi'=s_ns_{n+1}\ldots$ such that
  $\pi'\in\mathsf{DFTW}(p, \lambda)$ for some $\lambda\in\IN$.
  Every suffix of $\pi'$ satisfies $\mathsf{TGW}(p, \lambda)$. This implies
  that any odd priority in $\pi'$ is followed by a smaller even priority. It follows
  that the smallest priority appearing infinitely often in $\pi$ is even, as there
  are finitely many priorities.
\end{proof}

It can be shown that in some TAs, these inclusions may be strict. In other words,
the relations presented in Lemma~\ref{lemma:inclusions} are the most general
relationships for timed window parity objectives. We use the TA used to
show that fixed objectives are a strict refinement of parity objectives
in~\cite{MRS21} to exemplify this.

\begin{lemma}\label{lemma:strictness:inclusions}
  There exists a TA in which all time-divergent paths satisfy the parity
  objective, all of the inclusions of Lemma~\ref{lemma:inclusions}
  are strict and in which $\mathsf{FTW}(p,\lambda)\nsubseteq\mathsf{DBTW}(p)$
  and $\mathsf{DBTW}(p)\nsubseteq\mathsf{FTW}(p,\lambda)$ hold for any value
  of $\lambda$.
\end{lemma}
\begin{proof}
  Consider the TA $\mathcal{B}$ depicted in Figure~\ref{figure:paritynotwindow}
  and let $p_\mathcal{B}$ denote its priority function.
  It is easy to see that all time-divergent paths satisfy the parity objective:
  if the TA remains in location $\ell_1$ after some point, letting time diverge,
  the only priority seen infinitely often is $2$; if location $\ell_2$ is
  visited infinitely often, the smallest priority seen infinitely often is $0$.
  
  We will only consider sequences of states induced by initial paths of $\mathcal{B}$
  in the following arguments. We denote states by triples $(\ell, v, v')$ where
  $\ell\in\{\ell_0, \ell_1, \ell_2\}$ and $v$ and $v'$ respectively refer to the
  valuation of $x$ and of $\globalclock$.  Let $\lambda\in\IN$.

  \begin{figure}[ht]
    \centering
    \scalebox{0.8}{
      \begin{tikzpicture}[shorten <= 1pt, node distance=4cm, initial text=,
        every node/.style={transform shape},
        every state/.style={minimum size=1.5cm}]
        \node[state, initial, align=center] (l0) {$\ell_0$ \\ $x\leq 1$};
        \node[align=center, below of=l0, node distance=1.1cm] {$1$};
        \node[state, align=center, right of=l0] (l1) {$\ell_1$ \\ $\true$};
        \node[align=center, below of=l1, node distance=1.1cm] {$2$};
        \node[state, align=center, right of=l1] (l2) {$\ell_2$ \\ $x\leq 1$};
        \node[align=center, below of=l2, node distance=1.1cm] {$0$};
        \path[->] (l0) edge node[align=center, below] {$(\true, a, \varnothing)$} (l1);
        \path[->] (l1) edge node[align=center, below] {$(\true, a, \{x\})$} (l2);
        \path[->] (l2) edge[bend right] node[align=center, above] {$(\true, a, \{x\})$} (l0);
      \end{tikzpicture}}
        \caption{Timed automaton $\mathcal{B}$. Edges are labeled
    with triples guard-action-resets. Priorities are beneath
    locations. The incoming arrow with no origin indicates the initial location.} \label{figure:paritynotwindow}
  \end{figure}
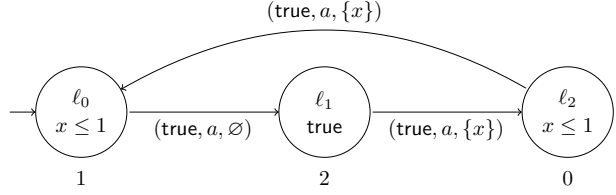

  Let us first show that
  $\mathsf{DFTW}(p, \lambda)\subsetneq \mathsf{FTW}(p,\lambda)$,
  $\mathsf{DBTW}(p_\mathcal{B})\subsetneq \mathsf{BTW}(p_\mathcal{B})$
  and $\mathsf{FTW}(p,\lambda)\nsubseteq\mathsf{DBTW}(p_\mathcal{B})$ hold.
  Due to the inclusions of Lemma~\ref{lemma:inclusions}, it suffices to
  provide some sequence of states in
  $\mathsf{FTW}(p_\mathcal{B}, \lambda)\setminus\mathsf{DBTW}(p_\mathcal{B})$
  to obtain these relations.
  For example, consider the sequence
  of states $(\ell_0, 0, 0)(\ell_1, 1, 1)(\ell_1, 2, 2)\ldots$ obtained by
  using action $a$ in $\ell_0$ after $1$ time unit, and then letting time
  diverge by means of delay moves.
  The suffix $(\ell_1, 1, 1)(\ell_1, 2, 2)\ldots$ of this
  sequence satisfies $\mathsf{DFTW}(p_\mathcal{B}, \lambda)$: the only priority
  that appears in this suffix is even. Therefore the sequence
  $(\ell_0,0,0)(\ell_1,1,1)(\ell_1,2,2)\ldots$ must satisfy
  $\mathsf{FTW}(p_\mathcal{B}, \lambda)$. However, this sequence does not satisfy
  $\mathsf{DBTW}(p_\mathcal{B})$ nor
  $\mathsf{DFTW}(p_\mathcal{B}, \lambda)$; no even priority smaller than $1$
  is ever seen,
  therefore there cannot be any good window at the start of the play.
 
  Let us now prove that $\mathsf{DFTW}(p_\mathcal{B}, \lambda)\subsetneq \mathsf{DBTW}(p_\mathcal{B})$,
  $\mathsf{FTW}(p_\mathcal{B}, \lambda)\subsetneq \mathsf{BTW}(p_\mathcal{B})$
  and $\mathsf{DBTW}(p_\mathcal{B})\nsubseteq\mathsf{FTW}(p_\mathcal{B}, \lambda)$ hold. It suffices to show
  that
  $\mathsf{DBTW}(p_\mathcal{B})\setminus\mathsf{FTW}(p_\mathcal{B}, \lambda)\neq\emptyset$. We provide a sequence satisfying
  $\mathsf{DFTW}(p_\mathcal{B}, \lambda + 1)\subseteq
  \mathsf{DBTW}(p_\mathcal{B})$ but not $\mathsf{FTW}(p_\mathcal{B}, \lambda)$.
  For instance, consider the sequence of states
  $(\ell_0, 0, 0)(\ell_1, 0, 0)(\ell_2, 0, \lambda)
  (\ell_0, 0, \lambda)(\ell_1, 0, \lambda)(\ell_2, 0, 2\lambda)\ldots$ obtained
  by repeatedly using action $a$ with a delay of $0$ in $\ell_0$,
  then action $a$ with a delay
  of $\lambda$ in $\ell_1$ and finally action $a$ in $\ell_2$ with a delay
  of $0$.
  The timed good window objective $\mathsf{TGW}(p, \lambda+1)$ is satisfied by every
  suffix of this sequence; there is a delay of $\lambda$ between an occurrence of the
  priority $1$ and the smaller even priority $0$. Therefore, the objective
  $\mathsf{DBTW}(p_\mathcal{B})$ is satisfied. However, the fixed objective
  $\mathsf{FTW}(p_\mathcal{B}, \lambda)$ is not satisfied; any suffix of this sequence starting
  in location $\ell_0$ does not satisfy the timed good window objective
  $\mathsf{TGW}(p_\mathcal{B}, \lambda)$ due to the delay spent in location $\ell_1$.
  The relations
  $\mathsf{DFTW}(p_\mathcal{B}, \lambda)\subsetneq \mathsf{DBTW}(p_\mathcal{B})$
  and
  $\mathsf{FTW}(p_\mathcal{B}, \lambda)\subsetneq \mathsf{BTW}(p_\mathcal{B})$
  follow from this example and inclusions of Lemma~\ref{lemma:inclusions}.

  It remains to show that $\mathsf{BTW}(p_\mathcal{B}) \subsetneq \mathsf{Parity}(p_\mathcal{B})$ holds.
  Initialize $n$ to $0$. We consider the sequence of states induced by the path
  obtained by sequentially using the
  moves $(0, a)$ in location $\ell_0$, $(n, a)$ in location $\ell_1$ and
  $(0, a)$ in location $\ell_2$, increasing $n$ and then repeating the procedure.
  This sequence of states satisfies the parity objective; the smallest priority
  seen infinitely often is $0$. However, it does not satisfy $\mathsf{BTW}(p)$.
  At each step of the construction of the path, a delay of $n$ takes place between
  priority $1$ in $\ell_0$ and priority $0$ in $\ell_2$. No matter the chosen suffix
  of the sequence of states and the chosen bound $\lambda$, the objective
  $\mathsf{DFTW}(p_\mathcal{B}, \lambda)$ cannot be satisfied, therefore the objective
  $\mathsf{BTW}(p_\mathcal{B})$ is not satisfied.
\end{proof}

\begin{remark}
  In the location $\ell_1$ of the previous TA, the invariant $\true$ allows
  us to wait for an arbitrary amount of time in $\ell_1$. However, this aspect
  of the TA is not crucial to illustrate that the inclusions of
  Lemma~\ref{lemma:inclusions} are strict.
  
  It is possible to obtain an example TA in which the claims of
  Lemma~\ref{lemma:strictness:inclusions} hold and such that invariants prevent
  time from diverging without infinitely often traversing edges. This can
  be accomplished by a straightforward adaptation of the TA $\mathcal{B}$
  of Figure~\ref{figure:paritynotwindow}; it suffices to change the invariant
  of $\ell_1$ to $x\leq 1$ and add an edge from $\ell_1$ to itself that
  resets $x$. Such an alteration does not change the behavior of the TA.
\end{remark}

The proof of Lemma~\ref{lemma:strictness:inclusions} illustrates that
window parity objectives, in general, are a strict strengthening of parity
objectives. Furthermore, it shows that, in general, there is no uniform
bound $\lambda$ such that all paths satisfying a direct or indirect bounded
timed window objective satisfy the corresponding fixed objective for $\lambda$.
However, we show in Section~\ref{section:verification} that if all
time-divergent paths of a TA satisfy a (respectively direct) bounded timed
window objective, then one can find a bound $\lambda$ such that all paths
satisfy the (respectively direct) fixed timed window
objective for $\lambda$ (Corollaries~\ref{corollary:direct:uniformity}
and~\ref{corollary:indirect:uniformity}).
Similarly, in TGs, $\player_1$ has a winning strategy for a bounded objective
if and only if they have a winning strategy for some corresponding fixed
objective (Theorems~\ref{theorem:games:direct:winning}
and~\ref{theorem:games:indirect:winning}).
In the sequel, we do not consider algorithms based on reductions to the
direct case; the bounds used to reduce bounded objectives to fixed objective
may be large and induce an otherwise avoidable computational cost.
This justifies alternative approaches.

\subsection{Simplifying paths violating window objectives}\label{section:objectives:technical}

In this section, we provide a technical result used for the verification and
realizability of bounded timed window objectives. First, we introduce some
terminology. We say a path $\pi=s_0\xrightarrow{m_0} s_1\ldots$ of $\automaton$
(respectively, a play $\pi=s_0 (m_0^{(1)}, m_0^{(2)})s_1\ldots$ of $\game$)
eventually follows the cycle
$(\ell_0, R_0)\xrightarrow{\tau}\ldots \xrightarrow{\tau} (\ell_n, R_n)$
of the region abstraction if there exists $i\in\IN$
such that for all $j\in \{i, i + 1, \ldots, i + n - 1\}$
and all $k\in\IN$, $[s_{j + n\cdot k}] = (\ell_j, R_j)$.
We say that a cycle of the region abstraction is \textit{time-divergent} if
all paths of the TA (or, equivalently, all plays of the TG) that eventually
follow this cycle are time-divergent.

The main result of this section allows us to extract
time-divergent cycles in the region abstraction from a path or play
violating a timed good window objective for a sufficiently large bound.
This result can then be applied to any path or play that violates the
direct or indirect bounded timed window objective; it follows from the
definition that,
for any bound $\lambda$, there is a suffix violating the timed good window
objective for $\lambda$.

For the sake of generality, we abstract whether we consider paths or plays:
we state the upcoming result in terms of sequences of states.
In practice, only the actions are abstracted away; delays
between states are encoded by the global clock $\globalclock$.
We say that for any sequence of states $s_0s_1\ldots\in S^\omega$, the
delays are bounded by $B\in\IN$ if
$v_{n+1}(\globalclock) - v_n(\globalclock)\leq B$ for all $n\in\IN$.
We also extend this
terminology to paths and plays via their induced sequence of states.

The rough idea of the following lemma is as follows:
assuming that delays are bounded
along a sequence of states, if the timed good window objective is violated for
some large enough
bound $\lambda\in\IN$, it is possible to find within $\lambda$ time units from the
start of the sequence a time-divergent cycle in the region abstraction.

In the context of TGs, we will seek to apply the result to construct an
outcome of a given finite-memory region strategy violating a window objective.
A \textit{deterministic finite automaton} (DFA) over state regions is a tuple
$(\fmstates, \mathfrak{m}_\init, \update)$ where $\fmstates$ is a finite set
of states, $\mathfrak{m}_\init\in\fmstates$ and
$\update\colon \fmstates\times (L\times\regions)\to \fmstates$.
Given a Mealy machine
$\fmmealy = (\fmstates, \mathfrak{m}_\init, \update, \nextmove)$
encoding a finite-memory region  strategy, we refer to
$(\fmstates, \mathfrak{m}_\init, \update)$ as the DFA underlying $\fmmealy$.
Showing that we can find cycles in the region abstraction gives us no
information on the finite-memory strategy. Therefore, we instead
require the stronger claim that we can find a cycle in the
product of the region abstraction of $\automaton$ and the underlying DFA
within the first $\lambda$ time units of the sequence of states.

\begin{lemma} \label{lemma:technical:bounded}
  Let $(\fmstates, \mathfrak{m}_\init, \update)$ be a DFA.
  Let $\pi = s_0s_1s_2\ldots\in S^\omega$ be a sequence of states induced by
  some time-divergent path or play in which delays are bounded by $1$,
  and $\mathfrak{m}_0\mathfrak{m}_1\mathfrak{m}_2\ldots\in\fmstates^\omega$
  be the sequence inductively defined by $\mathfrak{m}_0 = \mathfrak{m}_\init$
  and $\mathfrak{m}_{k+1} = \update(\mathfrak{m}_k, [s_k])$.
  Let $\lambda = 2\cdot |L|\cdot |\regions|\cdot |\fmstates|+ 3$.
  If $\pi\notin\mathsf{TGW}(p, \lambda)$, then there exist some
  indices $i < j$ such that
  $([s_i], \mathfrak{m}_i) = ([s_j], \mathfrak{m}_j)$, the
  global clock $\globalclock$ passes some integer bound between indices $i$ and $j$,
  and strictly less than $\lambda$ time units elapse before reaching $s_j$ from $s_0$.
\end{lemma}

\begin{proof}
  Assume that $\pi\notin\mathsf{TGW}(p, \lambda)$. For any $j\in\IN$, let
  $v_j$ denote the clock valuation of $s_j$. Because we assume that
  $\pi$ is induced by a path or play, the sequence
  $(v_j(\globalclock))_{j\in\IN}$ is non-decreasing. It follows that the
  set $\{j\in\IN \mid v_j(\globalclock) - v_0(\globalclock) < \lambda\}$
  is an interval.
  Let $j^\star$ denote the greatest element of this interval; $j^\star$ is well-defined
  because we assume that $\pi$ is induced by some time-divergent path or play.
  We let $h = s_0\ldots s_{j^\star}$ be the prefix of $\pi$
  in which strictly less than $\lambda$ time units have elapsed.
  Observe that because delays between states are at most of $1$
  and $v_{j^\star+1}(\globalclock) - v_0(\globalclock)\geq\lambda$, it follows
  that  $v_{j^\star}(\globalclock) - v_0(\globalclock) \geq \lambda - 1$.
  
  We find the sought indices $i$ and $j$ by progressively checking each
  index up to $j^\star$ by induction.
  We mark elements
  $([s_i], \mathfrak{m}_i)\in (L\times \regions)\times \fmstates$
  as unsuitable if at step $i$ of our search there is no $j > i$
  such that $([s_i], \mathfrak{m}_i) = ([s_j], \mathfrak{m}_j)$ and
  $\lfloor v_i(\globalclock)\rfloor < \lfloor v_j(\globalclock)\rfloor$
  (i.e., the global clock passes a new integer bound). If at any step
  we do not mark the current region as unsuitable, we have found the sought
  indices $i$ and $j$ and stop the search procedure.

  In the remainder of this proof, we show that the procedure must terminate
  by finding a suitable pair of indices. By contradiction, we assume that all
  elements of $(L\times \regions)\times \fmstates$ appearing in
  $([s_0], \mathfrak{m}_0)\ldots([s_{j^\star}], \mathfrak{m}_{j^\star})$
  are marked as unsuitable, i.e., the search of a suitable pair fails.
  
  Any $([s], \mathfrak{m})\in (L\times \regions)\times \fmstates$
  that is marked as unsuitable during the search procedure
  can only appear again at most one time unit after its first appearance, otherwise it
  would not have been marked as unsuitable.
  This implies that whenever a pair $([s], \mathfrak{m})$ is marked as
  unsuitable, there is some point fixed in time from its first appearance
  after which it no longer appears in $h$, i.e., there is
  some $\delay\in\IR_{\geq 0}$ (depending on the smallest index for which
  we witness the pair) such that for all $i\leq j^\star$,
  $v_i(\globalclock)\geq v_0(\globalclock) + \delay$ implies
  $([s_i], \mathfrak{m}_i)\neq ([s], \mathfrak{m})$, in which case we say that
  the pair $([s], \mathfrak{m})$ is eliminated by (as shorthand for can no
  longer appear from) time $\delay$.
  We give lower bounds on the number of eliminated pairs depending
  on the time that has passed. We reach a contradiction by showing that
  we run out of pairs in $(L\times \regions)\times \fmstates$ before we
  reach $j^\star$.

  We claim that at least $n$ pairs are eliminated by time $2n-1$.
  We prove this by induction. The base case is handled by considering
  the first elements of the sequence: $(s_0, \mathfrak{m}_0)$ is eliminated
  by time $1$. Let $k_1\leq j^\star$ denote
  the latest index such that
  $([s_0], \mathfrak{m}_0) = ([s_{k_1}], \mathfrak{m}_{k_1})$. This index
  occurs at most one time unit after index $0$.

  Now assume inductively that
  we have shown that (at least) $n$ distinct pairs
  $([s_{k_1}], \mathfrak{m}_{k_1})$, \ldots, $([s_{k_n}], \mathfrak{m}_{k_n})$
  (where $k_i\leq j^\star$ denotes the index of the last occurrence of a pair)
  are eliminated by time $2n-1$.
  It follows that $([s_{k_n+1}], \mathfrak{m}_{k_n+1})$ is eliminated at
  most $2$ time units after the elimination of
  $([s_{k_n+1}], \mathfrak{m}_{k_n+1})$: there is at most $1$ time unit between
  indices $k_n$ and $k_n +1$, and at most $1$ time unit between the first
  and last occurrence of $([s_{k_n+1}], \mathfrak{m}_{k_n+1})$. This shows that
  there are $n+1$ eliminated pairs by time $2n+1$.

  It follows that all elements of $(L\times \regions)\times \fmstates$
  are eliminated at time $\lambda - 2$, i.e., there are no more pairs
  that can appear in $h$ after this time.
  However, we have $v_{j^\star}(\globalclock) - v_0(\globalclock) \geq
  \lambda - 1$, i.e., it is absurd to have had
  $([s_{j^\star}], \mathfrak{m}_{j^\star})$ eliminated.
\end{proof}

The main interest of the lemma is to construct witness paths or plays that
violate the direct bounded timed window objective.
By following the sequence of states up to index $i$ and then looping in the
cycle formed by the sequence of states from $i$ to index $j$ (modulo
clock-equivalence), one obtains a path along which, at
all steps, the smallest priority seen from the start is odd, i.e., such that
no good window can ever be witnessed from the start.

\section{Verification of timed automata}\label{section:verification}

In this section, we are concerned with the verification of direct and indirect
bounded timed window objectives in TAs. For both objectives, we show
the equivalence of the following assertions:
(1) there exists a time-divergent witness to the violation of a (direct)
bounded objective, (2) there exists a time-divergent witness to the
violation of the matching (direct) fixed objective for a sufficiently
large bound, and (3) there exists a set of states (regions) reachable
from one another verifying some properties that we describe later in
this section.
Nondeterministic algorithms for the verification of the objectives are
obtained by guessing
appropriate regions and checking that they are reachable from one another.

The outline of the section is as follows.
Section~\ref{section:verification:criteria}
describes criteria attesting to the existence of time-divergent paths
violating the direct and indirect bounded timed window objectives in
timed automata. Verification algorithms are described
in Section~\ref{section:verification:complexity}.

We fix for this entire section a TA
$\automaton = (L, \ell_\init, C, \Sigma, I, E)$ and
a priority function $p\colon L\to \{0, \ldots, \maxpriority-1\}$.

\subsection{Equivalent conditions to the violation of bounded objectives}\label{section:verification:criteria}

In this section, we provide conditions equivalent to the existence of
paths violating the direct and indirect bounded timed window objectives.
We are also concerned with the question of uniformity of time bounds;
we show that in a timed automaton in which all time-divergent paths
satisfy a direct or indirect bounded timed window objective,
there exists a bound for which a direct or indirect fixed  objective is
satisfied.

\subsubsection{Direct bounded timed window objectives}

A path satisfies the direct bounded timed window objective if at all steps,
there is a good window and the size of these good windows is bounded
overall. Therefore, a path can violate this objective in one of two ways.
First, it may be the case that at some step, no good window of any size is
witnessed. Second, it may be the case that good windows are witnessed at
all steps, but that there is no bound on the size of these windows.

We show that whenever some time-divergent path violates the direct bounded
timed window objective, there is always some witness that falls in the first
category. Furthermore, a witness that takes the form of a path that
eventually follows a
time-divergent cycle of the region abstraction can be chosen.

Given a time-divergent path $\pi$ violating the direct bounded timed window
objective, the rough idea to derive a suitable witness is the following.
We consider some state $s_1$ along $\pi$ from which there is no good window
for the window size $\lambda$  in the statement of
Lemma~\ref{lemma:technical:bounded} (assuming that the DFA has only one state).
We obtain through this lemma two region-equivalent states $s_2$ and $s_2'$
appearing in $\pi$ within
$\lambda$ time units of $s_1$, such that in the path fragment of $\pi$ between
$s_2$ and $s_2'$, the global clock $\globalclock$ passes a new integer value. As
explained in Section~\ref{section:objectives:technical},
we can construct a path violating the direct bounded window objective
by following $\pi$ up to $s_2$ and then following a path that repeats the
time-divergent cycle in the region abstraction induced by
the sequence of states between $s_2$ and $s_2'$ in $\pi$.

The states described in the construction above can be characterized as follows.
First, there must be a finite path from the state $s_1$ to the state
$s_2$ in which the smallest priority in all prefixes is odd. Second, we require
that $[s_2]$ be reachable from $s_2$ without witnessing a good window from
$s_1$ and also in such a way that the global clock $\globalclock$ passes a new
integer bound. The latter property can also be translated to reachability
requirements; we require that, on the sought path from $s_2$ to its state
region, there be
states $s_3$ and $s_4$ such that the valuation of $\globalclock$ is integral
in only one of the two states $s_3$ and $s_4$.

We show hereunder that the existence of a time-divergent path violating
the direct bounded timed window objective is equivalent to the existence of
states satisfying the properties above. Because we need only consider a
state from which there is no good window for the bound $\lambda$ of
Lemma~\ref{lemma:technical:bounded}, these two conditions are also
equivalent to the existence of a time-divergent path violating
$\mathsf{DFTW}(p, \lambda)$.

\begin{theorem}\label{theorem:verification:direct}
  The following three statements are equivalent.
  \begin{enumerate}
  \item There exists a time-divergent initial path $\pi\notin\mathsf{DBTW}(p)$.
    \label{item:verification:direct:1}
  \item There exists a time-divergent initial path
    $\pi\notin\mathsf{DFTW}(p, 2\cdot|L|\cdot|\regions| + 3)$.
    \label{item:verification:direct:2}
  \item There exist reachable states $s_1$, $s_2$, $s_2'$,
  $s_3$ and $s_4$ such that $s_2\clockequiv s_2'$,
  the valuation of $\globalclock$ is integral in only one of the two
  states $s_3$ and $s_4$, and there is a finite path $h$ from
  $s_1$ to $s_2'$ passing through
  $s_2$, $s_3$ and $s_4$ in order such that the smallest priority in all
  of the prefixes of $h$ is odd.
  \label{item:verification:direct:3}
  \end{enumerate}
\end{theorem}

\begin{proof}
  Let $\lambda = 2\cdot|L|\cdot|\regions| + 3$. This $\lambda$
  is the bound of Lemma~\ref{lemma:technical:bounded} assuming a deterministic
  finite automaton with a single state.
  The implication
  (\ref{item:verification:direct:1}$\implies$\ref{item:verification:direct:2})
  follows directly from the inclusion
  $\mathsf{DFTW}(p, \lambda)\subseteq \mathsf{DBTW}(p)$.

  We move on to the proof of
  (\ref{item:verification:direct:2}$\implies$\ref{item:verification:direct:3}).
Let us assume there is some time-divergent initial path
$\pi\notin\mathsf{DFTW}(p, \lambda)$. We may assume without loss of generality
that delays in $\pi$ are bounded by $1$: moves $(\delay, a)$ with large delays can
be simulated by using $\lfloor \delay \rfloor$ moves of the form $(1, \bot)$
followed by the move $(\fracpart(\delay), a)$. It can easily
be shown that the objective is still violated following this modification.

Let $s_1$
be some state of $\pi$ such that some suffix $\pi'$ of $\pi$ starting in
$s_1$ violates the timed good window objective $\mathsf{TGW}(p, \lambda)$.
It follows from Lemma~\ref{lemma:technical:bounded} that there are two
region-equivalent states $s_2$ and $s_2'$ in $\pi'$ within the first
$\lambda$ time units such that, in $\pi'$, the global clock passes an integer
bound between the two states.

One can take $s_3 = s_2$. If in $s_2$, the valuation of $\globalclock$ is
an integer (resp.~not an integer), take $s_4$ to be any state in the path
from $s_2$ to $s_2'$ in $\pi'$ that has a non-integral (resp.~integral)
valuation. If no such state exists, it suffices to split a move
$(\delay, a)$ in $\pi'$ into two well-chosen  moves
$(\delay_1, \bot)$ and $(\delay_2, a)$ where $\delay = \delay_1 + \delay_2$
and the global clock $\globalclock$ is not equal to an
integer (respectively is equal to an integer) after $\delay_1$ time units elapse.
It follows from  $\pi'\notin\mathsf{TGW}(p, \lambda)$ that $s_1$, $s_2$, $s_2'$,
$s_3$ and $s_4$ satisfy the requirement of the theorem. This ends this direction
of the proof.

Finally, let us establish the implication (\ref{item:verification:direct:3}$\implies$\ref{item:verification:direct:1}).
Assume now that we have the states $s_1$, $s_2$, $s_2'$, $s_3$, $s_4$
and $h$ satisfying the properties in the statement of the theorem.
Let $h'$ denote the suffix of $h$ between states $s_2$ and $s_2'$.
We argue that any initial path $\pi$ obtained by reaching $s_1$ from $s_\init$
(by any means), then
following $h$ up to $s_2$, and then following the cycle in the region
abstraction induced by $h'$ is time-divergent and violates $\mathsf{DBTW}(p)$.
Fix one such path $\pi$ and let $\pi'$ denote its suffix starting from $s_1$.

First, let us argue the time-divergence of $\pi$. The path $\pi$
passes through states that are equivalent to $s_3$ and to $s_4$ infinitely often. In
other words, the global clock $\globalclock$ is infinitely often an integer, and
infinitely often not an integer. It follows from the fact that
$\globalclock$ cannot
be reset that it must pass infinitely many integer bounds, i.e., $\pi$ is
time-divergent.

Second, let us move on to showing that $\pi\notin\mathsf{DBTW}(p)$.
It suffices to show that no matter the bound $\lambda\in\IN$, the objective
$\mathsf{TGW}(p, \lambda)$ is violated by $\pi'$. This can be established by showing
that in any prefix of $\pi'$, the smallest priority that occurs is odd.
For any prefix of $\pi'$ that is a prefix of $h$, this property
follows from our hypothesis on $h$. For any subsequent prefix,
no new priorities are introduced as we repeat a cycle in the region abstraction
following
the suffix $h'$ of $h$. This shows that $\pi\notin\mathsf{DBTW}(p)$ and ends
the proof of this implication.

\end{proof}

In light of Theorem~\ref{theorem:verification:direct}, we directly
obtain the following corollary. 

\begin{corollary}\label{corollary:direct:uniformity}
  Let $\lambda = 2\cdot |L|\cdot |\regions| + 3$.
  All time-divergent paths of $\automaton$ satisfy $\mathsf{DBTW}(p)$ if and
  only if all time-divergent paths of
  $\automaton$ satisfy $\mathsf{DFTW}(p,\lambda)$.
\end{corollary}

Even though the corollary above suggests that we can reduce verification
of bounded objectives to verification of fixed objectives, the verification
of fixed objectives requires time polynomial in the supplied time bound.
Intuitively, one must explore the region abstraction of a TA derived from
$\automaton$ in which an additional clock $z\notin C$ is introduced and increases
up to the bound of the objective.
Given that the bound provided by Lemma~\ref{lemma:technical:bounded} is
large, we develop approaches that avoid the cost incurred
by this reduction.

\subsubsection{Bounded timed window objective}\label{section:verification:indirect}

We now move on to the bounded timed window objective. In this case,
time-divergent paths that eventually repeat a cycle of the region abstraction
no longer suffice as witnesses to the violation of the objective.
In the direct case, the finite path preceding the cycle mattered in the
violation, e.g., if an odd priority smaller than all those of the cycle
appeared along this path.
However, in such paths, only the cycle itself would
matter by prefix-independence for the indirect objective.

Lemma~\ref{lemma:strictness:inclusions} asserts the existence of a TA in
which all time-divergent paths satisfy the parity objective, but some
violate the bounded timed window objective. This implies that even if all
time-divergent cycles in the region abstraction have an even smallest
priority, this does not ensure the satisfaction of the bounded timed window
objective. It follows that the form of witnesses is more complex
in this case.

Nonetheless, witnesses can be always be found with a recursive structure.
Assume that some time-divergent path violates the bounded timed window
objective. Then there is a another violating path operating in stages
labeled by natural numbers $n\in\IN$, with each stage divided in two parts.
In the first part of stage $n$, we visit some well-chosen fixed state
region $[s]$
from which there is a time-divergent path that violates the direct bounded
objective. Once a state belonging to such a region is reached, we can follow a time-divergent path
violating the direct objective that eventually follows a cycle in the region
abstraction for (at least) $n$ time units, before moving on to stage $n+1$.

In the direct case, Theorem~\ref{theorem:verification:direct} essentially
states  that one can find a witness to the violation of the objective if
and only if there exists reachable states $s_1$, $s_2$ and $s_2'$ such that
$s_2\clockequiv s_2'$,
there is a finite path $h$ from $s_1$ to $s_2'$ passing through $s_2$ such
that the smallest priority in any prefix of $h$ is odd and an integer bound
is passed by the global clock between $s_2$ and $s_2'$ in $h$. The
characterization in the prefix-independent case is only slightly stronger:
we only require, in addition to the above, that $[s_1]$ be reachable from $s_2$,
without any constraints on the path between these two states. Intuitively,
we return to $[s_1]$ whenever a stage has ended. One such state is easy to
find: there are finitely many regions and infinitely many suffixes of the
path from which there are no good windows for some sufficiently large bound,
therefore some region must repeat.

We formalize our characterization below. Similarly to the direct case,
one can also show that the existence of a time divergent path violating
$\mathsf{BTW}(p)$ is equivalent to the existence of a
time-divergent path violating $\mathsf{FTW}(p, \lambda)$ for the bound
$\lambda$ of Lemma~\ref{lemma:technical:bounded}.

\begin{theorem}\label{theorem:verification:indirect}
  The following three statements are equivalent.
  \begin{enumerate}
  \item There exists a time-divergent initial path $\pi\notin\mathsf{BTW}(p)$.
    \label{item:verification:indirect:1}
  \item There exists a time-divergent initial path
    $\pi\notin\mathsf{FTW}(p, 2\cdot|L|\cdot|\regions| + 3)$.
    \label{item:verification:indirect:2}
  \item There exist reachable states $s_1$, $s_2$, $s_2'$,
  $s_3$ and $s_4$ such that $s_2\clockequiv s_2'$,
  the valuation of $\globalclock$ is integral in only
  of the two states $s_3$ and $s_4$,
  there is a finite path $h$ from $s_1$ to $s_2'$ passing through
  $s_2$, $s_3$ and $s_4$ (in order) such that the smallest priority in all
  of the prefixes of $h$ is odd, and the region $[s_1]$ is reachable from $s_2'$.
  \label{item:verification:indirect:3}
  \end{enumerate}
\end{theorem}

\begin{proof}
  Let $\lambda = 2\cdot|L|\cdot|\regions| + 3$
  be the bound of Lemma~\ref{lemma:technical:bounded} assuming a deterministic
  finite automaton with a single state.
  The implication (\ref{item:verification:indirect:1}$\implies$\ref{item:verification:indirect:2})
  follows directly from the inclusion
  $\mathsf{FTW}(p, \lambda)\subseteq \mathsf{BTW}(p)$.
  
  To establish the implication
  (\ref{item:verification:indirect:2}$\implies$\ref{item:verification:indirect:3}),
  we explain how to adapt the proof of
  Theorem~\ref{theorem:verification:direct} to derive the five states from
  a time-divergent initial path $\pi\notin \mathsf{FTW}(p, \lambda)$.   
  Let us assume that there is some time-divergent initial path 
  $\pi\notin\mathsf{FTW}(p, \lambda)$. In particular, 
  $\pi\notin\mathsf{DFTW}(p, \lambda)$. It is shown in the proof
  of Theorem~\ref{theorem:verification:direct} that by taking any state $s_1$
  in $\pi$ such that some suffix $\pi'$ of $\pi$ starting in
  $s_1$ violates the timed good window objective $\mathsf{TGW}(p, \lambda)$,
  we can find the sought-after states $s_2$, $s_2'$, $s_3$ and $s_4$,
  without the requirement that $[s_1]$ be reachable from $s_2'$.

  To ensure that $[s_1]$ is reachable from a matching $s_2'$, we 
  choose a state $s_1$ subject to some constraints. We show that
  there must be some state $s_1$ in $\pi$ such that some suffix
  $\pi'$ starting in $s_1$ satisfies $\pi'\notin \mathsf{TGW}(p, \lambda)$
  and such that there are states equivalent to $s_1$ infinitely often in $\pi'$.
  Because the region $[s_1]$ occurs
  infinitely often along $\pi$, there is an occurrence after
  the appearance of $s_2'$. This makes one such $s_1$ a good choice. The proof
  of existence of one such $s_1$ follows.
  
  Let $I = \{i\in \IN\mid \pi_{i\to}\notin \mathsf{TGW}(p, \lambda)\}$. The
  set $I$ must be infinite, otherwise there would be some $j\in\IN$ such
  that for all $i\geq j$, $\pi_{i\to}\in \mathsf{TGW}(p, \lambda)$, i.e.,
  $\pi_{j\to}\in\mathsf{DFTW}(p, \lambda)$, which would imply
  $\pi\in\mathsf{FTW}(p, \lambda)$. Because $I$ is infinite and there are finitely many
  state regions, one can find a state $s_1$ in $\pi$ indexed by an element in
  $I$ such that its state region is visited infinitely often along $\pi$.
  This ends the proof of this implication.

  Let us now move on to the implication (\ref{item:verification:indirect:3}$\implies$\ref{item:verification:indirect:1}). Assume the existence of states
  $s_1$, $s_2$, $s_2'$, $s_3$ and $s_4$ subject to the constraints above.
  We construct a time-divergent initial path $\pi\notin\mathsf{BTW}(p)$ inductively.
  We denote by $\pi_n$ the path constructed at step $n$ of the induction.
  We let $\pi_0$ be any finite path to $s_1$ from $s_\init$. Let us now assume that we are at
  induction step $n\geq 1$, and by induction that the last state of $\pi_{n-1}$
  is in $[s_1]$.
  We split the counterpart in the region abstraction of the path from
  $s_1$ through $s_2$, $s_3$, $s_4$ to $s_2'$ given by our hypothesis into two
  parts: $h_\regions$ for the part up to $[s_2]$ (not 
  included) and $h'_\regions$ for the remaining cycle from $[s_2]$ to itself. We
  extend $\pi_{n-1}$ by appending to it some path in $\automaton$ following the
  path $h_\regions(h_\regions')^{n+1}$ in the region abstraction, and then any
  path from $[s_2]$ back to $[s_1]$, so that we can continue the inductive
  construction.
  
  The path $\pi$ obtained through the inductive construction above 
  is time-divergent; the global clock, which cannot be reset, alternates between
  taking an integer value and not taking an integer value infinitely often, therefore its valuation must diverge.
    We now argue that $\pi$ violates $\mathsf{BTW}(p)$, i.e., we argue
  that for all suffixes of $\pi$, for all $\lambda\in\IN$,
  $\mathsf{DFTW}(p, \lambda)$ is not satisfied by the suffix.
  By construction, the path appended at step $n$
  of the construction aside from the return to $[s_1]$
  is such that, in all of its prefixes, the smallest priority is odd.
  Furthermore,
  the duration of this path is of at least $n$ time units: we witness the
  global clock pass an integer bound at least $n+1$ times in this path. It follows 
  that the suffix of $\pi$ after $\pi_{n-1}$ violates
  $\mathsf{TGW}(p, n)$. Because we let $n$ grow
  to infinity in the construction, no suffix of $\pi$ satisfies a direct
  fixed timed window objective. This ends the proof.
\end{proof}

In light of Theorem~\ref{theorem:verification:indirect}, we directly
obtain the following corollary.
\begin{corollary}\label{corollary:indirect:uniformity}
  Let $\lambda = 2\cdot |L|\cdot |\regions| + 3$.
  All time-divergent paths of $\automaton$ satisfy $\mathsf{BTW}(p)$ if and
  only if all time-divergent paths of $\automaton$ satisfy
  $\mathsf{FTW}(p,\lambda)$.
\end{corollary}

\subsection{Verification algorithms}\label{section:verification:complexity}

In this section, we discuss verification algorithms for the direct and indirect
bounded  objectives. In Section~\ref{section:verification:subreach}, we
provide a useful procedure to check, given some states of the TA, the existence
of paths subject to the constraints of
Theorems~\ref{theorem:verification:direct}
and~\ref{theorem:verification:indirect}.
We then discuss non-deterministic verification algorithms and their complexity
in Section~\ref{section:verification:algorithms}.

\subsubsection{Checking reachability with priority-induced constraints}\label{section:verification:subreach}
In the two previous sections, we have identified conditions for the existence
of paths violating the direct bounded timed window objectives and the bounded
timed window objectives. These criteria involve the existence of states such that
one can find a path traversing these states where, in any prefix of
this path, the smallest priority that occurs is odd, i.e., we construct
paths along which no good window is identified. We argue in the sequel 
that the existence of such paths can be decided in polynomial 
space.
We will outline a non-deterministic polynomial space procedure; the
previous claim follows from the equality
$\mathsf{PSPACE} = \mathsf{NPSPACE}$~\cite{DBLP:journals/jcss/Savitch70}.

This complexity can be justified by a straightforward adaptation of
the classical algorithm for reachability in timed automata~\cite{AlurD94}.
The idea is to detect a suitable path by means of the region abstraction.
The region abstraction itself is exponential in the size of the TA,
but needs not be constructed entirely to check whether some
region is reachable. An \textsf{NPSPACE} algorithm
for reachability can operate by exploring the region abstraction on-the-fly,
and keeping track of a region and the current number of steps taken in
the current path. The algorithm returns a positive answer if a target
is reached, and a negative answer if the step counter reaches the
size of the region abstraction. Because regions are representable in
polynomial space and the counter can be represented in binary, the claimed
complexity follows.

In the sequel, we require a slight variant of this algorithm. We are
given a certain number of regions $[s_1]$, \ldots, $[s_n]$ and want
to determine whether one path exists traversing these regions in
such a way that the smallest priority witnessed from the start of the
path is odd at all times. The classical algorithm can be extended
naturally to handle multiple sequential targets and the priority-related
constraints.

To handle the visiting of multiple regions in
order, it suffices, each time a target is reached, to reset the step counter
and update the target to the next one. One returns a 
positive answer if all targets have been reached. This induces an
increase in memory at most linear in the number of targets: one
can simply keep track of the current target by means of its index in
the sequence of targets. In practice, we use this procedure with five targets.

For the priority-related constraints, it suffices to keep track
of the smallest priority witnessed from the start of the guessed
path (unlike the counter above, this priority should
never be reset). We add an additional condition: the decision 
procedure stops and returns a negative answer if this priority
becomes even at any point. This induces an increase in memory of
at most $\log_2(d)$ bits. Overall, this modified procedure still
only uses polynomial space. We therefore obtain the following lemma.

\begin{lemma}\label{lemma:verification:reach}
  The existence of a path passing through $n$ given regions in
  order such that the smallest priority of all of its prefixes
  is odd is decidable in deterministic polynomial space.
\end{lemma}

\subsubsection{Algorithms for the verification of bounded timed window objectives}\label{section:verification:algorithms}

We can now describe the complexity of the verification problem
for direct and indirect bounded  timed window objectives.
We first describe algorithms for the dual problem of verification, i.e.,
algorithms that check whether there exists a time-divergent path that violates
the considered objective.
These algorithms use oracles to check reachability
properties between regions. The complexity of our algorithms is in
$\mathsf{NP}^\mathsf{PSPACE} = \mathsf{PSPACE}$~\cite{DBLP:journals/siamcomp/BakerGS75}. The idea is to guess five state regions and then check whether they conform to
the conditions in Theorems~\ref{theorem:verification:direct}
and~\ref{theorem:verification:indirect}.

We use two oracles in \textsf{PSPACE}. The first oracle returns, given two
regions, whether there is a path in the region abstraction from the first region to the second, i.e.,
this oracle decides
standard reachability. The second oracle encodes the problem formulated in
Lemma~\ref{lemma:verification:reach}.

To decide the existence of a time-divergent path violating the direct objective,
we guess five regions and check if they satisfy the conditions of
Theorem~\ref{theorem:verification:direct}. This algorithm consists
of guessing the regions, checking whether the first region is reachable from the
initial state using the first oracle and then using the second oracle to
confirm the satisfaction of conditions of
Theorem~\ref{theorem:verification:direct}.
For the indirect objective, we proceed similarly to check the conditions of
Theorem~\ref{theorem:verification:indirect}; the only difference to the
direct case is that there is an additional call to the first oracle.
This shows that the dual problem of verification for direct and indirect bounded
timed window objectives is in
$\mathsf{NP}^\mathsf{PSPACE} = \mathsf{PSPACE}$.
Because $\mathsf{PSPACE}$ is closed under complementation, the
$\mathsf{PSPACE}$-membership of the verification problem for direct and
indirect bounded timed window objectives follows.

\begin{lemma}\label{lemma:verification:complexity}
  The verification problems for direct and indirect bounded timed window
  objectives are in \textsf{PSPACE}.
\end{lemma}

Let us now assume that the priority function
$p\colon L\to\{0, \ldots, \maxpriority -1\}^\numdimensions$ is multi-dimensional.
Verifying that all time-divergent paths satisfy a generalized objective
is equivalent to checking that a one-dimensional objective is verified on
each dimension. It follows from Lemma~\ref{lemma:verification:complexity}
and $\mathsf{P}^\mathsf{PSPACE} = \mathsf{PSPACE}$~\cite{DBLP:journals/siamcomp/BakerGS75}
that the verification of multi-dimensional objectives can be done in polynomial
space.

\begin{theorem}\label{theorem:verification:complexity}
  The verification problems for generalized direct and indirect bounded
  timed window are in \textsf{PSPACE}.
\end{theorem}

\section{Solving timed games}\label{section:games}

In this section, we propose an algorithmic solution to the realizability
problem for direct and indirect bounded timed window parity objectives.
For the direct case, we provide a reduction to the realizability problem for
an $\omega$-regular region
objective in Section~\ref{section:games:direct}: we show that to enforce
the direct bounded objective, we can consider the objective requiring that
any odd priority is followed by a smaller even priority.
In Section~\ref{section:games:indirect}, we provide a fixed-point
algorithm for the indirect case, which intuitively iterates the computation
of a winning set for the direct case.

For this entire section, we fix a TG $\game = (\automaton, \Sigma_1, \Sigma_2)$
with $\automaton = (L, \ell_\init, C, \Sigma_1\cup\Sigma_2, I, E)$ and a
multi-dimensional priority function
$p\colon L\to \{0, \ldots, \maxpriority-1\}^\numdimensions$.

\subsection{Direct bounded timed window objective}\label{section:games:direct}
In this section, we provide a reduction from the realizability problem for
the generalized direct bounded timed window objective to the realizability
problem for the untimed $\omega$-regular request-response
objective~\cite{DBLP:conf/wia/WallmeierHT03,DBLP:conf/lata/ChatterjeeHH11}.
In Section~\ref{section:games:direct:rr}, we introduce the request-response
objective and explain how to derive a request-response objective from the
multi-dimensional priority function $p$. In
Section~\ref{section:games:direct:reduction}, we show that the set of winning states
for this request-response objective coincides with the winning set for the
generalized direct bounded timed window objective and that this set coincides
even with the winning set of some generalized direct fixed timed window
objective.

\subsubsection{Request response-objectives}\label{section:games:direct:rr}

A request-response objective is an $\omega$-regular region objective
defined by a family of pairs of sets of state regions
$\mathcal{R} = ((\mathsf{Rq}_j, \mathsf{Rp}_j))_{j=1}^r$. The request-response
objective for $\mathcal{R}$ requires that
for all $k\in\{1, \ldots, r\}$, for any visit to a state region in
$\mathsf{Rq}_k$, there must be a location in $\mathsf{Rp}_k$ appearing later in
the play. We refer to state regions in $\mathsf{Rq}_k$ as requests and to state
regions in $\mathsf{Rp}_k$ as responses.

Let $\mathcal{R} = ((\mathsf{Rq}_k, \mathsf{Rp}_k))_{k=1}^r$ be a family of
request-response pairs.
Formally, we define the request-response objective $\mathsf{RR}(\mathcal{R})$
as the set of sequences of states
\[\{s_0s_1\ldots \in S^\omega\mid
  \forall\, k\leq r,\,\forall n,\,\exists n'\geq n,\, [s_n]\in\mathsf{Rq}_k\implies [s_{n'}]\in\mathsf{Rp}_k \}.\]

A DPA in which the only priorities are $0$ and $1$
are equivalent to the deterministic Büchi automata (DBAs) of the literature.
In general, for a request-response objective with $r$ request-response
pairs, a DBA with $2^r\cdot r$ states suffices~\cite{DBLP:conf/wia/WallmeierHT03}. 
The request-response families we define later from
priority functions have $\lfloor \frac{\maxpriority}{2}\rfloor\cdot \numdimensions$
request-response pairs. Hence, using such a DBA in our game solving approach induces
an exponential blow-up in the number of priorities in the time complexity.
We can do better: the request-response objectives we derive from
multi-dimensional priority functions can be represented by DBAs with
$(\lfloor \frac{\maxpriority}{2}\rfloor+1)^\numdimensions\cdot\numdimensions$ states.
For a fixed number of dimensions, we obtain
a DBA with a number of states polynomial in the number of priorities.

We do not directly introduce small DBAs for the specific request-response
families derived from multi-dimensional priority functions. Instead, we
define a class of request-response families that subsumes them. We proceed
this way due to the indirect case; in the indirect case, we repeatedly
solve request-response games in which we alter the sets of requests and
responses; by introducing a broader class of
request-response families, we achieve a better complexity for these
computations with respect to the number of priorities.

We say that a family of request-response pairs
$\mathcal{R} = \{(\mathsf{Rq}_1, \mathsf{Rp}_1), \ldots, (\mathsf{Rq}_r, \mathsf{Rp}_r)\}$ is a
\textit{chain-response family} if the
sets of responses form a chain, i.e.,
$\mathsf{Rp}_1\supseteq \mathsf{Rp}_2 \supseteq\ldots \supseteq \mathsf{Rp}_r$
and each set of requests and responses are pairwise disjoint, i.e., for all
$i, j\leq r$, $\mathsf{Rq}_i\cap \mathsf{Rp}_j=\emptyset$. In a request-response
objective induced by a chain-response family, one needs only keep track of the
pending request with the fewest responses, because any response to this
request also addresses requests with more responses due to the chain of inclusions.
This allows us to define a DBA with $r + 1$ states; there is one state to
indicate that no requests are pending, and one state per
request to keep track of whichever pending request has fewest responses.

Let $\mathcal{R} = ((\mathsf{Rq}_k, \mathsf{Rp}_k))_{k=1}^r$ be a chain-response
family where $\mathsf{Rp}_1\supseteq \mathsf{Rp}_2 \supseteq\ldots \supseteq \mathsf{Rp}_r$. The request-response objective $\mathsf{RR}(\mathcal{R})$ can be
encoded by
a DBA $H = (\dpastates, q_\init, L\times\regions, \dpatransitions, p_H)$
where $\dpastates = \{0, 1, \ldots, r\}$, $q_\init = 0$, and $\dpatransitions$
is defined, for all $q\in\dpastates$ and $[s]\in L\times \regions$,
\[\dpatransitions(q, [s]) = \begin{cases}
    0 & \text{if } q\neq 0\text{ and }[s]\in\mathsf{Rp}_q \\
    \max(\{q\}\cup \{i\leq r\mid [s]\in\mathsf{Rq}_i \})
    & \text{otherwise,}
  \end{cases}\]
and the priority function $p_H$ assigns $0$ to state $0$ of $H$
and $1$ to all other states
of $H$. The DBA $H$ encodes $\mathsf{RR}(\mathcal{R})$. Indeed, $H$ keeps track
of the highest seen index of a request and a higher index means fewer
responses. Because request and response sets are pairwise disjoint, witnessing
the state $0$ of $H$ infinitely often is equivalent to having all requests
eventually answered.

We say a family of request-response pairs is an \textit{$n$-chain-response}
family if it is a union of $n$ chain-response families. Observe that for all
$\mathcal{R}_1$,\ldots, $\mathcal{R}_n$, we have
$\mathsf{RR}(\bigcup_{1\leq i\leq n}\mathcal{R}_i) =
\bigcap_{1\leq i\leq n}\mathsf{RR}(\mathcal{R}_i)$. The intersection of
the languages of $n$ DBAs with $r+1$ states
can be encoded by a DBA with $(r +1)^n\cdot n$
states~\cite[Proposition 6.1]{DBLP:books/daglib/0016866}.
It follows that request-response objectives obtained from
$n$-chain-response families where each underlying chain-response family has
at most $r$ pairs can be encoded by DBAs with at most $(r+1)^n\cdot n$
states.

The following result follows immediately from
Theorem~\ref{theorem:omegaregular:complexity}
and Theorem~\ref{theorem:fm:strategies}.

\begin{lemma}\label{lemma:games:rr}
  Let $\mathcal{R}$ be an $n$-chain-response family in which each underlying
  chain-response family has at most $r$ pairs.
  The set of winning states in $\game$ for the request-response objective
  $\mathsf{RR}(\mathcal{R})$ is a union of state regions and
  can be computed in time
  $\bigo((|L|\cdot|\regions|\cdot (r+1)^n\cdot n)^{3})$,
  and finite-memory region strategies proposing delays of at most $1$
  with $4\cdot (r+1)^n\cdot n$ states suffice for winning.
\end{lemma}

We now explain how we derive a $\numdimensions$-chain-response family from
a $\numdimensions$-dimensional priority function. The idea is to model each
odd priority on each dimension as a request, the responses to which are smaller
even priorities on the same dimension. A similar construction is used
for direct bounded objectives in games in
graphs~\cite{DBLP:journals/corr/BruyereHR16}.

Assume $p$ is a one-dimensional priority function.
We define the chain-response family $\mathcal{R}(p)$ as the
family of request-response pairs that contains for each odd priority
$j\in\{0, 1, \ldots, \maxpriority-1\}$,
the pair $(\mathsf{Rq}_{j}, \mathsf{Rp}_{j})$
where $\mathsf{Rq}_{j} = p^{-1}(j)\times\regions$ and $\mathsf{Rp}_{j} =
\{\ell\in L\mid p(\ell)\leq j \land p(\ell)\bmod 2 = 0\}\times\regions$.
This is indeed a chain-response family because the responses to an odd priority
are smaller even priorities, and are therefore also responses to any
greater odd priorities.
If $p$ is $\numdimensions$-dimensional, we let $\mathcal{R}(p)$ be the
$\numdimensions$-chain-response family
$\mathcal{R}(p)=\bigcup_{1\leq i\leq\numdimensions}\mathcal{R}(p_i)$.

We close this section by highlighting a nuance between the notion of good
windows and the modeling of priorities as requests and responses provided
in the definition of $\mathcal{R}(p)$.
We consider the one-dimensional case for the upcoming explanation.

Given a state occurring in a play, recall that
one finds a good window (of some size about which we are not concerned) if there
is a later state on the play such that the smallest priority seen between
the two states is even. The earliest response to a request in $\mathcal{R}(p)$
may not induce a good window; it may be the case that on the segment
between the request and response, we witness another odd priority for which
the first response is not suitable. This new priority must be strictly smaller
than that of the initial request; any response
to this new request is also strictly smaller than the first response. Assuming
that this second request is answered, there
may yet again be a strictly smaller odd priority between the second request
and response for which the second response is not suitable. We can
repeat this reasoning assuming the third request is answered, and so on. However, this
phenomenon can only occur finitely often due to the finite number of priorities.
The last response in the sequence of responses obtained above
is an even priority smaller than any prior odd priority, i.e., we witness a
good window eventually assuming that all requests are answered.

It follows that the request-response objective is satisfied if and only if
there are good windows from all states along the play. The remaining question
addressed in the following section is whether winning for the request-response
objective ensures the existence of a winning strategy for which the
size of these windows is bounded.

\subsubsection{Reducing direct objectives to request-response}\label{section:games:direct:reduction}
 
The goal of this section is to show that to solve the TG $\game$ with the
objective $\mathsf{GDBTW}(p)$, one can solve the TG $\game$ with the
request-response objective $\mathsf{RR}(\mathcal{R}(p))$.
The main argument consists in showing that the time-divergent outcomes of any
winning finite-memory region strategy for the objective
$\mathsf{RR}(\mathcal{R}(p))$ proposing bounded delays (the
existence of which is ensured by Lemma~\ref{lemma:games:rr} if
$\player_1$ wins) must satisfy $\mathsf{GDFTW}(p, \lambda)$ for some $\lambda\in\IN$. This implies that all time-divergent outcomes of one such strategy satisfy
$\mathsf{GDBTW}(p)$.

This result is shown by contradiction. We assume that there exists some
time-divergent outcome $\pi$ of one such finite-memory strategy violating
$\mathsf{DFTW}(p_k, \lambda)$ on some dimension $k$
for some sufficiently large bound $\lambda$;
this is ensured whenever one assumes the existence of a time-divergent outcome
of $\sigma$ violating $\mathsf{DBTW}(p_k)$.
If this is the case, we can construct an
outcome of $\sigma$ along which, on dimension $k$,
some odd priority is never
followed by a smaller even priority, i.e., some request goes unanswered, which
contradicts the fact that $\sigma$ is winning for
$\mathsf{RR}(\mathcal{R}(p))$.

The main points of the proof are as follows.
There is some suffix $\pi'$ of $\pi$ violating the timed good window objective
for $\lambda$. Within the $\lambda$ first time
units of $\pi'$, using Lemma~\ref{lemma:technical:bounded}, one can find
two indices such that the TG finds itself in state-equivalent states $s$ and
$s'$ and the Mealy machine encoding the winning strategy $\sigma$ finds itself in the
same memory states. Because we consider a finite-memory
region strategy, it is possible to inductively construct an outcome of
$\sigma$ which first follows $\pi$ up to $s$ and then
follows the time-divergent cycle in the region abstraction induced by $\pi$
between $s$ and
$s'$. However, because the smallest priority appearing in all prefixes of $\pi'$
up to $s'$ is odd (the timed good window objective is violated), it
follows that this specific priority is never followed by any smaller even
priority, contradicting the fact that $\sigma$ was winning
for $\mathsf{RR}(\mathcal{R}(p))$.

We provide the details hereunder. We prove a slightly stronger statement for
later use. Let $U$ be a set of state regions.
We show that the announced result holds even if we modify the
request-response pairs of $\mathcal{R}(p)$ by removing regions
in $U$ from all request sets and adding regions in $U$ to all response sets;
that is, any
time-divergent outcome of a finite-memory region winning strategy for the
modified request-response objective proposing
bounded delays satisfies some generalized direct fixed timed window objective,
and therefore the generalized direct bounded timed window
objective, under the assumption that the regions in $U$ are not visited.

\begin{lemma}\label{lemma:games:technical}
  Let $\mathcal{R}(p) = (\mathsf{Rq}_i, \mathsf{Rp}_i)_{i= 1}^r$
  be the family of request-response pairs derived from $p$ and
  $\mathcal{R}' = (\mathsf{Rq}_i\setminus U, \mathsf{Rp}_i\cup U)_{i= 1}^r$
  for some set of state regions $U\subseteq L\times\regions$.
  Let $W$ denote the set of winning states for the objective
  $\mathsf{RR}(\mathcal{R}')$ and
  $\fmmealy = (\fmstates, \mathfrak{m}_\init, \update, \nextmove)$
  be a Mealy machine encoding a finite-memory region winning strategy of
  $\player_1$ from $W$ for the objective $\mathsf{RR}(\mathcal{R}')$
  proposing delays of at most $1$.
  Let $\pi = s_0 (m_0^{(1)},m_0^{(2)}) s_1\ldots$ be a time-divergent outcome
  of the strategy induced by $\fmstates$ such that $s_0\in W$ and let
  $\lambda = 2 \cdot |L|\cdot |\regions|\cdot |\fmstates| + 3$.
  If for all $n\in\IN$, $[s_n]\notin U$ then
  $\pi\in \mathsf{GDFTW}(p, \lambda)\subseteq \mathsf{GDBTW}(p)$.
\end{lemma}

\begin{proof}
  Assume that $\pi\notin\mathsf{GDFTW}(p, \lambda)$ by contradiction.
  We fix a dimension $k\in\{1, \ldots, \numdimensions\}$ such that
  $\pi\notin\mathsf{DFTW}(p_k, \lambda)$.
  Let $\sigma$ denote
  the strategy induced by $\fmmealy$. We consider the sequence
  $\mathfrak{m}_0\mathfrak{m}_1\ldots\in\fmstates^\omega$ of memory states
  witnessed along $\pi$, given by $\mathfrak{m}_0 = \mathfrak{m}_\init$ and
  for all $n\in\IN$, $\mathfrak{m}_{n+1} = \update(\mathfrak{m}_{n}, [s_n])$.
  
  Recall that $\lambda$ is the bound
  of Lemma~\ref{lemma:technical:bounded} using the DFA underlying $\fmmealy$.
  It follows from our assumption of
  $\pi\notin\mathsf{DFTW}(p_k, \lambda)$ that there is some $n_0\in\IN$ such that
  $\pi_{n_0\to}\notin\mathsf{TGW}(p_k, \lambda)$. By
  Lemma~\ref{lemma:technical:bounded}, there exists two indices
  $n_1$, $n_2\geq n_0$ such that $n_1 < n_2$, $s_{n_1}\clockequiv s_{n_2}$,
  $\mathfrak{m}_{n_1} = \mathfrak{m}_{n_2}$ and
  $\lfloor v_{n_1}(\globalclock)\rfloor < \lfloor v_{n_2}(\globalclock)\rfloor$,
  and for all $n_0\leq n'\leq n_2$, we have $\min_{n_0\leq n\leq n'}p(\ell_n)$
  is odd, where $s_n = (\ell_n, v_n)$ for all $n\in\IN$.

{ \newcommand{\altmove}{\widetilde{m}}
  \newcommand{\altstate}{\widetilde{s}}
  \newcommand{\altmemory}{\widetilde{\mathfrak{m}}}
  \newcommand{\altplay}{\widetilde{\pi}}
  We now construct a time-divergent outcome
  $\altplay=\altstate_0(\altmove_0^{(1)}, \altmove_0^{(2)})\altstate_1\ldots$
  of $\sigma$ that does not satisfy the request-response objective
  $\mathsf{RR}(\mathcal{R}')$. We denote by
  $\altmemory_0\altmemory_1\ldots\in\fmstates^\omega$ the sequence of memory
  states along the play $\altplay$.
  We define $\altplay_{|n_2} =\pi_{|n_2}$, i.e., the play $\altplay$ coincides
  with $\pi$ up
  to step $n_2$. It follows that for any $n\leq n_2$, we have
  $\altmemory_n = \mathfrak{m}_n$. In particular, we have
  $\altmemory_{n_2} = \mathfrak{m}_{n_1}$.
  
  The remainder of the construction is by induction.
  Let $k\in\IN$ and $j\in \{0, \ldots, n_2 - n_1 - 1\}$. We will choose
  $\altstate_{n_2 + (n_2 - n_1)\cdot k + j}$ such that it is equivalent to
  $s_{n_1 + j}$ and
  $\altmemory_{n_2 + (n_2 - n_1)\cdot k + j} = \mathfrak{m}_{n_1+j}$. The idea
  to extend $\altplay$ is to follow
  the cycle in the region abstraction induced by the history
  $s_{n_1}(m_{n_1}^{(1)}, m_{n_1}^{(2)})\ldots(m_{n_2-1}^{(1)}, m_{n_2-1}^{(2)})s_{n_2}$.
  In practice, to ensure time-divergence of the constructed play, we
  ensure that the moves $\altmove_{n_2 + (n_2 - n_1)\cdot k + j}^{(1)}$ and
  $\altmove_{n_2 + (n_2 - n_1)\cdot k + j}^{(2)}$ are such that
  $\widetilde{\delay} = \delayfunc(m_{n_2 + (n_2 - n_1)\cdot k + j}^{(1)}, m_{n_2 + (n_2 - n_1)\cdot k + j}^{(2)})$
  traverses the same regions from $\altstate_{n_2 + (n_2 - n_1)\cdot k + j}$
  than $\delay = \delayfunc(m_{n_1 + j}^{(1)}, m_{n_1 + j}^{(2)})$ does from
  $s_{n_1+j}$,
  i.e.,
  $\{[v_{n_1+j} + \delay_{\mathsf{mid}}] \mid
  0\leq\delay_{\mathsf{mid}}\leq \delay\}
  =
  \{[\widetilde{v} + \delay_{\mathsf{mid}}] \mid 0\leq\delay_{\mathsf{mid}}\leq
  \widetilde{\delay} \}$
  where $v_{n_1+j}$ and $\widetilde{v}$ denote the clock valuations
  in $s_{n_1 + j}$ and in $\altstate_{n_2 + (n_2 - n_1)\cdot k + j}$ respectively.

  We only provide the construction of
  $\altmove^{(1)}_{n_2}$, $\altmove^{(2)}_{n_2}$, $\altstate_{n_2+1}$
  and $\altmemory_{n_2+1}$ (i.e., case $k = 0$ and $j=0$) for the sake
  of readability. Other cases are handled similarly.
  We define $\altmove^{(1)}_{n_2} = \nextmove(\altmemory_{n_2}, \altstate_{n_2})
  = \nextmove(\mathfrak{m}_{n_1}, \altstate_{n_2})$
  to ensure consistency of $\altplay$ with $\sigma$. To
  define $\altmove^{(2)}_{n_2}$, we distinguish two cases depending on which
  player is responsible for the transition in $\pi$ at step $n_1$.

  Assume first that $s_{n_1}\xrightarrow{m_{n_1}^{(1)}}s_{n_1+1}$ holds.
  Then we let $\altmove^{(1)}_{n_2}$ be
  any $\player_2$ move enabled in $\altstate_{n_2}$ with a delay greater than
  or equal to that of $\altmove_{n_1}^{(1)}$.
  Let $\altstate_{n_2+1}$ be the unique state such that
  $\altstate_{n_2}\xrightarrow{\altmove_{n_1}^{(1)}}\altstate_{n_2+1}$ holds.
  Because $\sigma$ is a finite-memory region strategy, the equivalence
  $\altstate_{n_2+1}\clockequiv s_{n_1+1}$ is ensured and the same regions
  are traversed from $s_{n_1}$ and $\altstate_{n_2 +1}$ in $\pi$ and
  $\altplay$ respectively.
  It follows from
  $\mathfrak{m}_{n_1}=\altmemory_{n_2}$ and $s_{n_1}\clockequiv\altstate_{n_2}$
  that $\mathfrak{m}_{n_1+1}=\altmemory_{n_2+1}$.
  This closes this case of the inductive
  step.

  Now, assume that $s_{n_1}\xrightarrow{m_{n_1}^{(1)}}s_{n_1+1}$ does not hold.
  In this case, the move of $\player_2$ is responsible for the transition at
  step $n_1$ in $\pi$. Let $\delay=\delayfunc(m_{n_1}^{(2)})$ and
  let $v_{n_1}$ and $\widetilde{v}_{n_2}$ denote the clock valuations in state
  $s_{n_1}$ and $\altstate_{n_2}$ respectively.
  We choose $\altmove_{n_2}^{(2)}=(\widetilde{\delay}, \action(m_{n_1}^{(2)}))$
  for some $\widetilde{\delay}\leq \delayfunc(\altmove_{n_1}^{(1)})$ such that
  $\widetilde{v}_{n_2} + \widetilde{\delay} \in [v_{n_1} + \delay]$
  and $\{[v_{n+j} + \delay_{\mathsf{mid}}] \mid
  0\leq\delay_{\mathsf{mid}}\leq \delay\} =
  \{[\widetilde{v}_{n_2} + \delay_{\mathsf{mid}}] \mid
  0\leq\delay_{\mathsf{mid}}\leq \widetilde{\delay} \}$; one such delay exists
  because the moves $m_{n_1}^{(1)}$ and $\altmove_{n_2}^{(1)}$ traverse the same
  regions from $s_{n_1}$ and $\altstate_{n_2}$ respectively ($\sigma$
  is a finite-memory region strategy), and the region
  $[v_{n_1} + \delay]$ has the region $[v_{n_1} + \delayfunc(m_{n_1}^{(1)})]$
  as a successor.
 
  Let $\altstate_{n_2+1}$ be the unique state such that
  $\altstate_{n_2}\xrightarrow{\altmove_{n_2}^{(2)}}\altstate_{n_2+1}$.
  By choice of $\altmove_{n_2}^{(2)}$, we have
  $\altstate_{n_2+1}\in\jointdestination(\altstate_{n_2}, \altmove_{n_2}^{(1)}, \altmove_{n_2}^{(2)})$.
  Furthermore, because guard satisfaction is uniform within a region
  and resets preserve regions, it follows that
  $\altstate_{n_2+1}\clockequiv s_{n_1+1}$.
  Finally, we must have
  $\mathfrak{m}_{n_1+1}=\altmemory_{n_2+1}$ for the same reason as in the
  previous case.

  We now argue that $\altplay$ is time-divergent and does not satisfy the
  request-response objective $\mathsf{RR}(\mathcal{R}')$. Time-divergence
  follows from the fact  that the global clock $\gamma$ passes an integer bound
  between indices
  $n_1$ and $n_2$ in $\pi$ and that all regions traversed between these
  indices in $\pi$ are infinitely often traversed in $\altplay$.
  For the request-response objective, we first remark
  that for all $n\in\IN$, $[\altstate_n]\notin U$, because all states appearing
  in $\altplay$ are equivalent to states in $\pi$.
  Hence, requests and responses along $\altplay$ are determined
  by the sequence of witnessed locations and their priorities. Let
  $n^\star\in\argmin_{n_0\leq n\leq n_2}p_k(\ell_n)$. From index $n^\star$ in
  $\altplay$, no priority smaller than $p_k(\ell_{n^\star})$ appears
  on dimension $k$, and this
  priority is odd. This shows that some request goes unanswered in $\altplay$,
  i.e., $\altplay\notin\mathsf{RR}(\mathcal{R}')$, contradicting the fact that
  $\sigma$ is winning.

}
\end{proof}

\begin{remark}
  The proof of Lemma~\ref{lemma:games:technical} can be adapted to show that
  if a state is winning for an arbitrary request-response objective
  $\mathsf{RR}(\mathcal{R})$, then it is winning for a bounded variant thereof,
  in which we require that the delay between requests and responses
  along a play be bounded by some integer.
\end{remark}

It follows from Lemma~\ref{lemma:games:technical} that any state winning for
the objective $\mathsf{RR}(\mathcal{R}(p))$ is also winning for some generalized
direct fixed timed window objective, thus for the generalized direct bounded
timed window objective. A consequence is that one can use the synthesis
algorithm for games with request-response objectives to construct winning
strategies for the generalized direct bounded timed window objective from
these states.

It remains to argue that states that are winning for the generalized bounded window
objective are also winning for the request-response objective. This follows
immediately
from the inclusion $\mathsf{GDBTW}(p)\subseteq\mathsf{RR}(\mathcal{R}(p))$: if a
play satisfies $\mathsf{GDBTW}(p)$, there must be good windows (of bounded
size) at all times and on all dimensions along a play, implying that all odd
priorities are always followed by smaller even priorities.
We obtain the following result.

\begin{theorem}\label{theorem:games:direct:winning}
  Let $\lambda = 8 \cdot |L|\cdot |\regions|\cdot
  (\lfloor\frac{\maxpriority}{2}\rfloor +1)^\numdimensions\cdot
  \numdimensions + 3$.
  The sets of winning states for the objectives
  $\mathsf{GDFTW}(p, \lambda)$,
  $\mathsf{GDBTW}(p)$ and $\mathsf{RR}(\mathcal{R}(p))$ coincide.
  Furthermore, there exists a finite-memory region strategy that is winning
  for all three objectives from any state in these sets.
\end{theorem}

\begin{proof}
  We first argue that from the set of winning states for
  $\mathsf{RR}(\mathcal{R}(p))$, there exists a strategy winning for all
  three objectives at once. Lemma~\ref{lemma:games:rr} ensures that there
  exists a finite-memory region strategy $\sigma^\fmmealy$ induced by a Mealy machine
  with $4\cdot (\lfloor\frac{\maxpriority}{2}\rfloor +1)^\numdimensions\cdot
  \numdimensions$ states and proposing delays of at most $1$ suffices to
  win for $\mathsf{RR}(\mathcal{R}(p))$.
  It follows from Lemma~\ref{lemma:games:technical} that $\sigma^\fmmealy$
  is winning for the objectives
  $\mathsf{GDFTW}(p, \lambda)$ and $\mathsf{GDBTW}(p)$
  from any state from which $\player_1$ has a winning strategy for
  $\mathsf{RR}(\mathcal{R}(p))$.

  It follows from the above that the set of winning states for
  $\mathsf{RR}(\mathcal{R}(p))$ is a subset of the set of winning states of
  the two window objectives. Furthermore, the inclusion
  $\mathsf{GDFTW}(p, \lambda) \subseteq \mathsf{GDBTW}(p)$
  implies that any state winning for the fixed objective is also winning for
  the bounded objective.
  To end the proof, it suffices to show that the
  inclusion $\mathsf{GDBTW}(p)\subseteq\mathsf{RR}(\mathcal{R}(p))$ holds
  to obtain that the set of winning states for the bounded window objective
  is included in that of the request-response objective.

  Let $\pi = s_0(m_0^{(1)}, m_0^{(2)})s_1\ldots\in\mathsf{GDBTW}(p)$ be a
  play conforming to the generalized direct bounded timed window objective.
  Let $n\in\IN$ such that $s_n\in\mathsf{Rq}$ for some
  $(\mathsf{Rq}, \mathsf{Rp})\in\mathcal{R}(p)$. There is some dimension
  $k\in\{1, \ldots, \numdimensions\}$ such that
  $\mathsf{Rq} = p^{-1}_k(j)\times \regions$ for some odd priority $j$.
  By definition, there is some
  $\lambda\in\IN$ such that $\pi\in\mathsf{DFTW}(p_k, \lambda)$, which
  implies $\pi_{n\to}\in\mathsf{TGW}(p_k, \lambda)$. It follows immediately
  from the definition of $\mathsf{TGW}(p_k, \lambda)$ that there exists some
  $n' > n$ such that the priority of the location of $s_{n'}$ on dimension $k$
  is even and smaller than $j$, i.e., $[s_{n'}]\in\mathsf{Rp}$. This shows
  that $\pi$ satisfies the request-response objective, and ends the proof.
\end{proof}

We conclude this section by determining the time complexity of solving the
realizability problem for direct bounded timed window objectives.
We produce a request-response objective with
$\numdimensions\cdot\lfloor\frac{\maxpriority}{2}\rfloor$ pairs, i.e., our reduction is in
polynomial time. In light of Lemma~\ref{lemma:games:rr}, we obtain that
the overall reduction-based algorithm described above for realizability in TGs
with direct bounded timed window objectives is in exponential time.

\begin{theorem}\label{theorem:games:direct:complexity}
  The realizability problem for TGs with generalized direct bounded timed
  window objectives is in \textsf{EXPTIME}.
\end{theorem}
\begin{proof}
  It takes time
  $\bigo((\lfloor\frac{\maxpriority}{2}\rfloor + 1)^\numdimensions
  \cdot\numdimensions\cdot |L|\cdot|\regions|)$ to construct a
  DBA encoding $\mathsf{RR}(\mathcal{R}(p))$ (the factor $|L|\cdot|\regions|$
  comes from the construction of transitions), and by
  Lemma~\ref{lemma:games:rr},
  it takes time
  $\bigo((|L|\cdot|\regions|\cdot
  (\lfloor\frac{\maxpriority}{2}\rfloor + 1)^\numdimensions \cdot\numdimensions))^{3})$
  to solve the request-response game. Overall, we need exponential time to
  solve the game.
\end{proof}

\subsection{Indirect bounded timed window objective}\label{section:games:indirect}

In this section, we show the \textsf{EXPTIME}-membership of the realizability
problem for the generalized bounded timed window objective.
To this end, we provide a fixed-point algorithm to solve these games.
At each step of the algorithm, we compute the set of winning states for a
given request-response objective.

The structure of the section is as follows. We open the section by
presenting the algorithm and proving its termination in
Section~\ref{section:games:indirect:algorithm}.
The correctness of the algorithm is shown in
Section~\ref{section:games:indirect:correctness}.
Section~\ref{section:games:indirect:complexity}
establishes that the algorithm terminates in exponential time.

\subsubsection{An algorithm for solving bounded timed window games}\label{section:games:indirect:algorithm}

We provide a fixed-point algorithm to compute the set of winning states for
the bounded timed window objective. We utilize request-response objectives
as in the direct case.

The algorithm behaves as follows. We start by computing the winning set
$W^1$ for the direct objective via the request-response objective
$\mathsf{RR}(\mathcal{R}(p))$; we obtain in this way a subset of the set of
winning states, because
$\mathsf{GDBTW}(p)\subseteq \mathsf{GBTW}(p)$.
It follows from the prefix-independence of $\mathsf{GBTW}(p)$ that $\player_1$
can extend any play that reaches $W^1$ into a winning play. Hence, we can
compute a larger subset $W^2$ of the set of winning states of $\player_1$ by
changing our request-response pairs in such a way that reaching $W^1$ clears
all requests.

This set $W^2$ is a subset of the set of winning states; intuitively
if $\player_1$ uses a winning strategy for the simplified request-response
objective from $W^2$ and the play does not reach $W^1$, the outcome satisfies
the winning condition for the direct objective $\mathsf{GDBTW}(p)$
by Lemma~\ref{lemma:games:technical} with $U=\{[s]\mid s\in W^1\}$.
This reasoning can be repeated inductively: we update the request-response
objective so that states in $W^2$ clear all requests. We continue until a fixed
point is reached; the set of states $W$ obtained this way is a set of states
from which $\player_1$ has a winning strategy for the objective
$\mathsf{GBTW}(p)$.

We now formally present the algorithm.
The steps of the algorithm are as follows. First, we construct the family of
request-response pairs $\mathcal{R}(p)$. After this initialization,
the algorithm enters a loop, in which we repeatedly solve request-response
games. We modify the request-response
pairs at each step by marking regions that were in the latest computed
winning set as responses for all possible requests. Note that at each step,
we always have $\numdimensions$-chain-response families.
The algorithm terminates when the set of winning states no longer grows.
The procedure is summarized in Algorithm~\ref{algorithm:indirect:game}, in
which we assume a sub-routine $\mathsf{SolveRR}$ which given a TG and a
$\numdimensions$-chain-response family $\mathcal{R}$, outputs the set of
winning states in the TG for
the objective $\mathsf{RR}(\mathcal{R})$.

\begin{algorithm}[h]
  \caption{Computing the set of winning states for $\mathsf{BTW}(p)$} \label{algorithm:indirect:game}
  \KwData{A TG $\game = (\automaton, \Sigma_1, \Sigma_2)$, a multi-dimensional priority function $p$ over $\automaton$.}
  $k\leftarrow 0$\;
  $W^0\leftarrow \emptyset$\;
  $\mathcal{R}\leftarrow \mathcal{R}(p)$\;
  \Repeat{$W^{k}\setminus W^{k-1}=\emptyset$}
  {
    $k\leftarrow k+1$\;
    $W^k\leftarrow \mathsf{SolveRR}(\game, \mathcal{R})$\;
    \For{$(\mathsf{Rq}, \mathsf{Rp})\in\mathcal{R}$}
    {
      $\mathsf{Rq}\leftarrow \mathsf{Rq}\setminus \{[s]\in L\times\regions \mid [s]\subseteq W^k\}$\;
      $\mathsf{Rp}\leftarrow \mathsf{Rp} \cup\{[s]\in L\times\regions \mid [s]\subseteq W^k\}$\;
    }
  }
  \Return $W^k$\;
\end{algorithm}

We now move on to the termination of Algorithm~\ref{algorithm:indirect:game}.
It is known that the set of winning states for $\omega$-regular region
objectives in TGs are unions of state regions (Theorem~\ref{theorem:omegaregular:complexity}).
Hence, it suffices to show
that the sequence of sets $(W^k)_{k\in K}$ computed at each step of the
algorithm is non-decreasing to obtain a proof of termination, as there are
finitely many state regions. Let us note that it is due to this property that
we refer to Algorithm~\ref{algorithm:indirect:game} as a fixed-point algorithm.
Intuitively, the result holds because we simplify the request-response
objectives from one iteration to the next.

\begin{lemma}\label{lemma:indirect:termination}
  The sequence of sets $(W^k)_{k\in K}$ computed in the loop of
  Algorithm~\ref{algorithm:indirect:game} is non-decreasing. As a consequence,
  Algorithm~\ref{algorithm:indirect:game} terminates.
\end{lemma}
\begin{proof}
  We show the first statement of the lemma.  We proceed by induction. We
  trivially have $W^0\subseteq W^1$ given that $W^0=\emptyset$. Let us
  now take $k\in K$, $k< \sup K$, and show that $W^k\subseteq W^{k+1}$.
  Let $\mathcal{R}^k$ and $\mathcal{R}^{k+1}$ respectively
  denote the family of request-response pairs from which $W^k$ and
  $W^{k+1}$ were computed.
To obtain $W^k\subseteq W^{k+1}$, it suffices to show that 
  $\mathsf{RR}(\mathcal{R}^{k})\subseteq \mathsf{RR}(\mathcal{R}^{k+1})$;
  $W^k$ and $W^{k+1}$ are the respective winning sets for these objectives.

  Let $s_0s_1\ldots\in \mathsf{RR}(\mathcal{R}^k)$
  be an infinite sequence of states. We must show that
  $s_0s_1\ldots\in \mathsf{RR}(\mathcal{R}^{k+1})$.
  Let $(\mathsf{Rq}^{k+1}, \mathsf{Rp}^{k+1})\in\mathcal{R}^{k+1}$ be a
  request-response pair. Assume there exists $i\in\IN$ be such
  that $s_i\in\mathsf{Rq}^{k+1}$.
  It follows from the innermost loop of the algorithm that there is some
  request-response pair
  $(\mathsf{Rq}^k, \mathsf{Rp}^k)\in\mathcal{R}^k$ such that
  $\mathsf{Rq}^{k+1} \subseteq \mathsf{Rq}^k$ and
  $\mathsf{Rp}^{k+1} \supseteq \mathsf{Rp}^k$. It follows from
  $s_0s_1\ldots\in \mathsf{RR}(\mathcal{R}^{k})$ and
  $\mathsf{Rq}^{k+1} \subseteq \mathsf{Rq}^k$ that there is some $j\geq i$
  such that $s_j\in \mathsf{Rp}^k \subseteq \mathsf{Rp}^{k+1}$. This shows that
  $s_0s_1\ldots\in\mathsf{RR}(\mathcal{R}^{k+1})$. This ends the argument that
  $(W^k)_{k\in K}$ is non-decreasing.

  It remains to show that Algorithm~\ref{algorithm:indirect:game} terminates.
  Each $W^k$, $k\in K$, is a union of state regions. There are finitely
  many state regions and we have shown $(W^k)_{k\in K}$ to be non-decreasing,
  thus it follows the sequence eventually reaches a fixed point, i.e., the
  algorithm terminates.
\end{proof}

\subsubsection{Correctness of the fixed-point algorithm}\label{section:games:indirect:correctness}

In this section, we prove that the set $W$ returned by
Algorithm~\ref{algorithm:indirect:game} is the set of winning states for
$\player_1$ in the TG $\game$ for the objective $\mathsf{GBTW}(p)$. We
establish the stronger claim that $\player_1$ has a strategy that is winning
for some generalized fixed timed window objective from $W$.

The proof is done in two steps. First, we show that
Algorithm~\ref{algorithm:indirect:game} outputs a subset of the set of
winning states
of $\player_1$ on which finite-memory region strategies suffice.
Second, to end the proof of correctness, we show that the
complement of the returned set is not winning for $\mathsf{GBTW}(p)$.

Let us argue that $\player_1$ has a winning strategy from any state in
the set $W$ returned by Algorithm~\ref{algorithm:indirect:game}.
The set $W$ is organized in layers: each set
$W^k\setminus W^{k-1}$ is one such layer.
We can construct winning strategies by exploiting this layered structure.
In the lowermost layer $W^1$, we have the winning set
for $\mathsf{GDBTW}(p)$, which is also winning for a fixed objective
(Theorem~\ref{theorem:games:direct:winning}); any winning strategy for the
direct objective is trivially winning for the indirect objective.

Higher layers are handled inductively. Given some layer, e.g.,
$W^k\setminus W^{k-1}$, one argues that $\player_1$ wins by constructing a
strategy that changes its behavior when a lower layer is reached: as long as
the layer does not change, $\player_1$ plays a winning strategy for the
current request-response objective, and should a deeper layer be reached,
$\player_1$ wins by forgetting the history and switching to a winning
strategy in this deeper layer. All outcomes of this strategy are
winning by prefix-independence of the objective; once the layer index
no longer decreases, Lemma~\ref{lemma:games:technical} ensures
that that some generalized direct fixed objective is satisfied if time diverges.

By choosing finite-memory region winning strategies for each request-response
objective in the construction of the layered winning strategy, we can even
show that finite-memory region strategies suffice for winning in $W$.
The idea is to keep
track of the current layer in memory, and whenever the layer of the current
state is lower than that in the memory, we act as though we had just
started the play in the current state.
The following lemma and its proof formalize the explanations above.

\begin{lemma}\label{lemma:indirect:game:inclusion}
  Let $\lambda = 8 \cdot |L|\cdot |\regions|\cdot
  (\lfloor\frac{\maxpriority}{2}\rfloor +1)^\numdimensions\cdot
  \numdimensions + 3$.
  The set $W$ provided by Algorithm~\ref{algorithm:indirect:game}
  is a subset of the set of winning states of $\player_1$
  for the objective $\mathsf{GFTW}(p, \lambda)$
  and finite-memory region strategies suffice for winning from any state
  in $W$. 
\end{lemma}
\begin{proof}
  We first describe a Mealy machine encoding a winning finite-memory region
  strategy of $\player_1$, and then prove it indeed encodes a winning strategy.

  Let $K$ denote the set of positive integers such that $W^k\setminus W^{k-1}$
  is non-empty. For each $k\in K$, let $\mathcal{R}^{k}$ be the family of
  request-response pairs from which $W^{k}$ was computed, and let
  $\fmmealy^k = (\fmstates, \mathfrak{m}_\init, \update^k, \nextmove^k)$
  be a Mealy machine encoding a finite-memory region strategy for $\player_1$
  on $W^k$ for the objective $\mathsf{RR}(\mathcal{R}^k)$, proposing delays
  of at most $1$ (Lemma~\ref{lemma:games:rr}) with
  $4\cdot (\lfloor \frac{\maxpriority}{2}\rfloor +1)^\numdimensions\cdot\numdimensions$ states.
  We can assume that these Mealy machines all share the same state space
  $\fmstates$. For any state $s\in W$, we
  let $\mathsf{e}([s]) = \min\{k\in K\mid [s]\subseteq W^{k}\}$ denote the earliest
  index $k\in K$ such that $s\in W^k$.
  
  We will consider the Mealy machine
  $\fmmealy = (\fmstates\times K, (\mathfrak{m}_\init,  \max K),
  \update, \nextmove)$ where the update function
  $\update\colon \fmstates\times K\times \regions\to\fmstates$
  is defined, for all $\mathfrak{m}\in\fmstates$, $k\in K$ and $s\in S$, by
  \[\update(\mathfrak{m}, k, [s]) = \begin{cases}
      (\update^k(\mathfrak{m}, [s]), k)& \text{if $s\notin W$ or $\mathsf{e}([s]) \geq k$} \\
      (\update^{\mathsf{e}([s])}(\mathfrak{m}_\init, [s]), \mathsf{e}([s])) & \text{otherwise},
    \end{cases}\]
  and $\nextmove\colon \fmstates\times K\times \regions\to M_1$ is
  defined, for all $\mathfrak{m}\in\fmstates$, $k\in K$ and $s\in S$, by
    \[\nextmove(\mathfrak{m}, k, [s]) = \begin{cases}
      \nextmove^k(\mathfrak{m}, [s]) & \text{if $s\notin W$ or $\mathsf{e}([s]) \geq k$} \\
      \nextmove^{\mathsf{e}([s])}(\mathfrak{m}_\init, [s]) & \text{otherwise}.
    \end{cases}\]
  Intuitively, $\fmmealy$ encodes a strategy that plays a winning strategy
  of $W^k$ as long as the play remains in $W^k$, and whenever the plays visits
  a state in some
  $W^{k'}$ with $k'< k$, forgets the past and switches to a winning strategy
  in $W^{k'}$.

  We now show that $\fmmealy$ encodes a strategy that is winning from
  every state in $W$.
Let $\pi = s_0(m_0^{(1)}, m_0^{(2)})s_2\ldots$ be an
  outcome of the strategy induced by $\fmmealy$ starting in some state of
  $W$, and let $((\mathfrak{m}_n, k_n))_{n\in \IN}$ be the sequence of memory
  states such that $(\mathfrak{m}_0, k_0) = (\mathfrak{m}_\init, \max K)$
  and for all $n\in\IN$, $ (\mathfrak{m}_{n+1}, k_{n+1}) = \update(\mathfrak{m}_n, k_n, [s_n])$.
  We first argue that there exists some $k\in K$ such that $\pi$ has a
  suffix that starts in $W^{k}\setminus W^{k-1}$ and that is consistent with
  the strategy induced by $\fmmealy^k$.

  By construction of $\update$, the sequence $(k_n)_{n\in\IN}$
  is a non-increasing sequence of non-negative integers, hence it must stabilize at some
  point to some $k\in K$.
  We now argue that $\pi$ has a suffix starting in $W^k\setminus W^{k-1}$
  and consistent with the strategy induced by $\fmmealy^k$.
  We distinguish two cases, depending on whether $k= k_0$ or not.

  First, assume that $k < k_0 = \max K$.
  Let $n_0 = \min\{n\in\IN\mid k_n = k\} - 1$.
  By definition of $\update$, it must be the case that $\mathsf{e}([s_{n_0}]) = k$,
  i.e., $s_{n_0} \in W^k\setminus W^{k-1}$. The suffix
  $\pi_{n_0\to}$ is consistent with the strategy encoded by $\fmmealy^k$:
  the first move
  in $\pi_{n_0\to}$ is given by $\nextmove^k(\mathfrak{m}_\init, [s_{n_0}])$
  by definition of $\nextmove$, and for later steps, it follows
  from the fact that $(k_n)_{n > n_0}$ is a constant sequence
  and the definitions of $\update$ and $\nextmove$. Indeed, memory
  updates and move proposals are performed as they would be in
  $\fmmealy^k$: the component in $K$ of memory states of $\fmmealy$ is
  disregarded.

  Now, assume that $k = k_0 = \max K$. It follows from the definition of
  $\nextmove$ that $\pi$ is consistent with the strategy induced by
  $\fmmealy^k$. Furthermore, by definition
  of $\update$, it must be the case that no state in $W^{k-1}$ has been
  visited, i.e., $\pi$ starts in $W^k\setminus W^{k-1}$. In this case, $\pi$
  itself is a play consistent with $\fmmealy^k$ that starts in
  $W^k\setminus W^{k-1}$. We let $n_0 = 0$ so as to treat both cases
  simultaneously in the remainder of the proof; note that $\pi = \pi_{0\to}$.

  We now prove that
  $\pi\in\wc_1(\mathsf{GFTWP}(p, \lambda))$.
  The definition of the indirect window objectives imply that it suffices to
  show that
  $\pi_{n_0\to}\in \wc_1(\mathsf{GDFTWP}(p, \lambda))$.
  If $\pi_{n_0\to}$ is time-convergent, then it must be blameless for $\player_1$
  because it is the outcome of a winning strategy (the strategy induced by
  $\fmmealy^k$).
  Assume that $\pi_{n_0\to}$ is time-divergent. By
  Lemma~\ref{lemma:indirect:termination}, we obtain
  $W^{k-1} = \bigcup_{k'\leq k-1} W^{k'}$. We can therefore apply
  Lemma~\ref{lemma:games:technical} with $U = \{[s]\mid s\in W^{k-1}\}$ and $\fmmealy^k$
  to obtain that $\pi_{n_0\to}$ satisfies $\mathsf{GDFTW}(p, \lambda))$. 
  We have shown that the strategy induced by $\fmmealy$ is winning from
  every state of $W$, ending the proof.
\end{proof}

Lemma~\ref{lemma:indirect:game:inclusion} asserts that all states in the
output $W$ of Algorithm~\ref{algorithm:indirect:game} are winning for some
generalized fixed timed window objective, hence for the generalized
bounded timed window objective. To finish
the proof of correctness of the algorithm, it remains to show that states
outside of $W$ are not winning for the generalized bounded timed window
objective. The idea is to show that for any strategy of
$\player_1$, there is some losing outcome when starting from $S\setminus W$.

We use the fact that states in $S\setminus W$ are losing for some
request-response objective where states in $W$ answer all pending requests.
It follows that any losing time-divergent outcome eventually stays in
$S\setminus W$. We can inductively construct an outcome that violates the
bounded timed window objective in stages as follows. At stage $n$,
we forget about the past and follow a play that is losing for the
request-response objective, while remaining consistent with the strategy of
$\player_1$ fixed beforehand. We follow this play until some request is left
pending for at least $n$ time units, and then move on to the next stage.
This constructs some losing outcome of $\player_1$'s strategy because requests
come from odd priorities: an unanswered request for $n$ time units implies
the existence of a window that is not good within $n$ time units on some dimension.
Because there are finitely many dimensions, this outcome cannot satisfy the
bounded timed window objective on at least one dimension,
and it follows that $\player_1$'s strategy is not winning. We formalize this
construction in the following proof.

\begin{lemma}\label{lemma:indirect:game:complement}
  Let $W$ denote the set provided by Algorithm~\ref{algorithm:indirect:game}.
  From every state in $S\setminus W$, $\player_1$ has no winning strategy
  for the objective $\mathsf{GBTW}(p)$.
\end{lemma}

\begin{proof}
  Let $\mathcal{R}$ be the request-response family
  $\{(\mathsf{Rq}\setminus [W], \mathsf{Rp}\cup [W])\mid
  (\mathsf{Rq}, \mathsf{Rp})\in \mathcal{R}(p)\}$ appearing in the last
  iteration of Algorithm~\ref{algorithm:indirect:game}, where
  $[W] = \{[s]\in L\times\regions\mid s\in W\}$ denotes
  the set of state regions in $W$.
  The set $W$ is the set of winning states of $\player_1$ for the request-response
  objective $\mathsf{RR}(\mathcal{R})$.
  It follows that from each state in $S\setminus W$, $\player_1$ has no
  winning strategy for this objective.

  Let $s\in S\setminus W$. We must show that $\player_1$ has no winning
  strategy for $\mathsf{GBTW}(p)$ from $s$. Fix a strategy $\sigma$ of
  $\player_1$. If $\sigma$ has some time-convergent outcome that is not
  blameless for $\player_1$, $\sigma$ cannot be a winning strategy. In the sequel,
  we assume that all time-convergent outcomes of $\sigma$ are blameless for $\player_1$.

  We construct an outcome of $\sigma$ from $s$ by inductively extending
  histories. We will denote the history after step $n\in\IN$ of the construction
  by $h_n$. The inductive assumptions we rely on is that $h_n$ ends in some
  state of $S\setminus W$, and that $h_n$ is consistent with $\sigma$.
  The play obtained in the limit of the construction will be an outcome of
  $\sigma$ from $s$. To ensure that this outcome violates $\mathsf{GBTW}(p)$,
  we construct our histories as follows: for $n\geq 1$, in the history
  appended to $h_{n-1}$ so as to obtain $h_n$,
  on some dimension, there is some odd priority that was not
  followed by any smaller even priority within $n$ time units.
  The resulting play violates $\mathsf{GBTW}(p)$: no matter the suffix
  taken, even if all windows along it are good on all dimensions, there is
  no bound on their size over all dimensions, hence there is some dimension
  on which there is no bound.
  Furthermore, this play is also time-divergent; at step $n$ of the
  construction, a history in which at least $n$ time units elapse is appended.

  We introduce some notation. Given a history
  $h = s_0(m_0^{(1)}, m_0^{(2)})\ldots (m_{k-1}^{(1)}, m_{k-1}^{(2)}) s_k\in\hist(\game)$ and
  some history or play
  $\pi = s_k(m_k^{(1)}, m_k^{(2)})\ldots (m_{k'-1}^{(1)}, m_{k'-1}^{(2)}) s_{k'}\ldots
  \in\hist(\game)\cup\plays(\game)$ where the last state of $h$ is the
  first state of $\pi$, we let $h\cdot\pi$ denote the play or history
  $s_0(m_0^{(1)}, m_0^{(2)})\ldots(m_{k-1}^{(1)}, m_{k-1}^{(2)})
  s_k(m_k^{(1)}, m_k^{(2)})\ldots(m_{k'-1}^{(1)}, m_{k'-1}^{(2)}) s_{k'}\ldots$
  obtained by concatenating $h$ and $\pi$ and disregarding the repeated state.

  We let $h_0 = s$, which trivially satisfies the inductive assumptions.
  Now, assume that $h_n$ has been constructed, and let $s_n=\last(h_n)$.
  We consider a strategy $\sigma_n$ such that for any history
  $h\in\hist(\game)$ starting in $s_n$, $\sigma_n(h) = \sigma(h_n\cdot h)$.
  The strategy $\sigma_n$ uses the actions which would have been proposed by
  $\sigma$ if we had seen $h_n$ prior to the input history.
  
  Because $s_n$ is losing for $\mathsf{RR}(\mathcal{R})$, there exists
  some outcome $\pi_n$ of $\sigma_n$ starting in $s_n$ such that
  $\pi_n\notin\wc_1(\mathsf{RR}(\mathcal{R}))$. By choice of $\sigma_n$,
  the play $h_n\cdot \pi_n$ is an outcome of $\sigma$.
  The play $\pi_n$ must be time-divergent, otherwise the play $h_n\cdot \pi_n$
  would not be blameless for $\player_1$, which contradicts
  the assumption that all time-convergent outcomes of $\sigma$ are
  blameless for $\player_1$. We therefore have that $\pi_n$ is time-divergent
  and $\pi_n\notin\mathsf{RR}(\mathcal{R})$.

  There is some request along $\pi_n$ that is never followed by a response.
  From states in $W$, all requests are answered, and hence there are no
  occurrences of $W$ after this request.
  Furthermore, because $\pi_n$ is time-divergent, there is some
  prefix $h(\pi_n)$ of $\pi_n$ such that at least $n+1$ time units elapse
  between the unanswered request and the last state of $h(\pi_n)$. It suffices
  to choose $h_{n+1} = h_n\cdot h(\pi_n)$ to obtain all desired properties.

  This concludes the construction of a time-divergent outcome of $\sigma$ that
  violates $\mathsf{GBTW}(p)$. We have thus shown that $\sigma$ is not a winning
  strategy, and therefore that $s$ is not in the set of winning states of $\player_1$
  for $\mathsf{GBTW}(p)$.
\end{proof}

We summarize the results of the section in the following theorem.
On the one hand, Lemma~\ref{lemma:indirect:game:inclusion} implies that the
set returned by Algorithm~\ref{algorithm:indirect:game} is a subset of the
set of winning states for $\mathsf{GBTW}(p)$ on which finite-memory region
strategies suffice, and such a strategy can be chosen as winning for some
generalized fixed timed window objective. On the other hand,
Lemma~\ref{lemma:indirect:game:complement} implies that $\player_1$ has no
winning strategy on the complement of the set computed by
Algorithm~\ref{algorithm:indirect:game} for any fixed timed window objective.
We obtain the following theorem.

\begin{theorem}\label{theorem:games:indirect:winning}
  Let $\lambda = 8 \cdot |L|\cdot |\regions|\cdot
  (\lfloor\frac{\maxpriority}{2}\rfloor +1)^\numdimensions\cdot
  \numdimensions + 3$.
  The sets of winning states for the objectives
  $\mathsf{GFTW}(p, \lambda)$ and $\mathsf{GBTW}(p)$ coincide.
  Furthermore, there exists a finite-memory region strategy that is winning
  for both objectives from any state in these sets of winning states.
\end{theorem}

\subsubsection{Complexity of the fixed-point algorithm}\label{section:games:indirect:complexity}

We conclude this section by determining the computational complexity of
Algorithm~\ref{algorithm:indirect:game}. In the worst case, there are
as many iterations as there are state regions. While the complexity of
solving the request-response games is in \textsf{EXPTIME}, we still obtain an
\textsf{EXPTIME} algorithm because the exponential terms are multiplied rather
than stacked.
We obtain the following result.

\begin{theorem}\label{theorem:games:indirect:complexity}
  The realizability problem for bounded timed window objectives is
  in \textsf{EXPTIME}.
\end{theorem}
\begin{proof}
  We show that
  Algorithm~\ref{algorithm:indirect:game} runs in exponential time to
  finish this proof.
  By Lemma~\ref{lemma:games:rr} the subroutine $\mathsf{SolveRR}(\mathcal{R})$
  runs in time
  $\bigo((|L|\cdot|\regions|\cdot (\lfloor\frac{\maxpriority}{2}\rfloor +1)^\numdimensions\cdot\numdimensions)^{3})$ (the time
  to solve the TG dominates that of the DPA construction). The innermost
  loop iterates
  $\lfloor\frac{\maxpriority}{2}\rfloor\cdot\numdimensions$ times.
  The outermost loop iterates at most
  $|L|\cdot |\regions|$ times. By combining all of these complexities
  appropriately, one obtains a time complexity in
  $\bigo((|L|\cdot|\regions|\cdot (\lfloor\frac{\maxpriority}{2}\rfloor +1)^\numdimensions\cdot\numdimensions)^{4})$.
\end{proof}

\section{Lower bounds and completeness}\label{section:completeness}
In this section, we establish the \textsf{PSPACE} and \textsf{EXPTIME} completeness of the
verification and realizability problems for the direct and indirect bounded
timed window objectives. In light of
Theorems~\ref{theorem:verification:complexity},~\ref{theorem:games:direct:complexity} and~\ref{theorem:games:indirect:complexity} that assert membership
of these problems in these complexity classes,
we need only establish hardness to obtain completeness.
In the remainder of this section, we no longer distinguish direct and indirect
cases; arguments are the same in both cases.

We will consider the verification and realizability problems for safety
objectives to establish hardness. The verification problem for safety objectives
is \textsf{PSPACE}-complete~\cite{AlurD94},
and the realizability problem for safety objectives is
\textsf{EXPTIME}-complete (as a consequence of the \textsf{EXPTIME}-completeness
of the safety control problem~\cite{HenzingerK99}).

It was shown in~\cite{MRS21} that there exists a polynomial-time reduction from
the verification and realizability problems for safety objectives to these
respective problems for fixed timed window objectives. Furthermore, the
reduction works no matter the bound on the size of windows in the definition
of the fixed objective, i.e., the two problems for safety objectives are
reducible in polynomial time to their counterpart for the objectives
$\mathsf{FDTW}(p, \lambda)$ or $\mathsf{FTW}(p, \lambda)$ for any
$\lambda\geq 1$ (for some appropriate priority function $p$) with a
construction independent of $\lambda$.
We argue that these same reductions are suitable to establish hardness of the
studied problems with bounded timed window objectives.

An intuitive sketch of the reductions follows.
They are similar for both the verification and realizability problems; we
modify the TA in the same way in both cases, and make no changes to the
partition of actions in TGs.
Let $\automaton = (L, \ell_\init, C, \Sigma, I, E)$ be a TA. Fix $F\subseteq L$.
Recall that the safety objective for $F$ requires that no location of $F$ ever
be visited. 

The reduction consists in deriving a TA $\automaton'$ from $\automaton$
in which locations are augmented with a Boolean value indicating whether $F$ has been
previously visited. Edges of $\automaton$ are replicated in $\automaton'$. These
edges do not update the Boolean value, unless they target some location in $F$,
in which case the Boolean value is changed to indicate $F$ has been visited.
The initial location $(\ell_\init, b)$
of $\automaton'$ indicates that $F$ has been visited if and only if
$\ell_\init\in F$. To define the window objectives, we use a priority function
assigning $0$ (respectively $1$) to locations indicating $F$ has not been visited
(respectively has been visited). Intuitively, correctness of the reduction follows from
the fact that if $F$ is never visited, then only the priority $0$ appears,
and otherwise, from some point on, only the priority $1$ appears. In the former
case, any variant of timed window parity objectives are satisfied trivially,
and in the latter, they are trivially violated.

Formally, we can also derive hardness for the verification and realizability problems
for the bounded timed window objectives as follows.
We have established that the verification and realizability problem for
bounded timed window objectives are equivalent to some instance of verification
and realizability problems respectively for 
some fixed timed window objective on the same TA or TG (Corollaries~\ref{corollary:direct:uniformity}
and~\ref{corollary:indirect:uniformity} for verification and
Theorems~\ref{theorem:games:direct:winning}
and~\ref{theorem:games:indirect:winning} for realizability).
Because the reduction above is known to work for fixed objectives for any
bound $\lambda$, it follows that the verification and realizability problems
for safety objectives are reducible in polynomial time to
the verification and realizability problems for the bounded timed window
objectives, yielding \textsf{PSPACE} and \textsf{EXPTIME}-hardness of these
problems in the one-dimensional case. We obtain the following result.

\begin{theorem}\label{theorem:complexity:completeness}
  The verification problem for generalized direct and indirect bounded
  timed window objectives is \textsf{PSPACE}-complete and
  the realizability problem for generalized direct and indirect bounded
  timed window objectives is \textsf{EXPTIME}-complete.
\end{theorem}

\section{Comparing window objectives in timed and untimed settings}\label{section:comparison}

In this section, we provide a short comparison of timed window objectives.
We compare the timed and untimed settings, as well as the fixed and bounded
settings.
A summary of the complexity classes for each respective problem is provided
in Table~\ref{table:complexity}. We fix a TG
$\game = (\automaton, \Sigma_1, \Sigma_2)$ with
$\automaton = (L, \ell_\init, C, \Sigma_1\cup\Sigma_2, I, E)$ for the upcoming
explanations.

\begin{table}[htb]
  \begin{center}
    \begin{tabular}{| c | c | c | c |}
      \hline
      & & Single dimension & Multiple dimensions \\ \hline
      \multirow{2}*{Timed automata} & Fixed~\cite{MRS21}
      & \textsf{PSPACE}-complete & \textsf{PSPACE}-complete \\ \cline{2-4}
      ~ & Bounded
      & \textbf{\textsf{PSPACE}-complete} & \textbf{\textsf{PSPACE}-complete} \\ \hline
      \multirow{2}*{Timed games} & Fixed~\cite{MRS21}
      & \textsf{EXPTIME}-complete & \textsf{EXPTIME}-complete \\ \cline{2-4}
      ~ & Bounded
      & \textbf{\textsf{EXPTIME}-complete} & \textbf{\textsf{EXPTIME}-complete} \\ \hline
      \multirow{2}*{Games (untimed)~\cite{DBLP:journals/corr/BruyereHR16}} & Fixed
      & \textsf{P}-complete & \textsf{EXPTIME}-complete \\ \cline{2-4}
      ~ & Bounded
      & \textsf{P}-complete & \textsf{EXPTIME}-complete \\ \hline
    \end{tabular}
    \caption{Summary of the complexity classes for problems with
      window parity objective in timed and untimed settings.
      Direct and prefix-independent cases are grouped together as their
      complexity matches. New results are in boldface.}\label{table:complexity}
  \end{center}
\end{table}

First, let us compare TGs with parity objectives and with window parity
objectives by analogy to the setting of untimed games.
In the one-dimensional case, in both the fixed and bounded cases, solving
untimed games with window parity objectives can be done in polynomial time.
Parity games on graphs are widely studied and have recently been shown to
be solvable in quasi-polynomial time~\cite{CaludeJKL017}, but are not yet known
to be solvable in polynomial time. In many algorithms, the number of priorities
is responsible for their high complexity. One-dimensional window parity
games provide a polynomial time alternative to parity games; the number of
priorities contributes polynomially to the complexity of solving an untimed
window parity game.

In the timed setting, TGs with parity objectives can be solved in exponential
time~\cite{AlfaroFHMS03}, and~\cite{ChatterjeeHP11} proposes a reduction
from parity TGs to untimed parity games; from a TG
 and a priority function $p\colon L\to\{0, \ldots, \maxpriority -1\}$, they construct a turn-based parity
game with $256\cdot |L|\cdot |\regions|\cdot |C|\cdot\maxpriority$ states
and priorities at most $\maxpriority + 1$. The solving of parity TGs by means
of this reduction nevertheless suffers
from the blow-up in the number of priorities in the same way as untimed games.
Similarly to the untimed setting, fixed and bounded timed window objectives avoid this issue; the number of
priorities only contributes polynomially to the complexity of solving these
games.

Now let us move on to a comparison of the fixed and bounded cases.
Despite there being no difference in the complexity classes for the two
cases, a TG with a generalized direct or indirect fixed timed window objective
with $\numdimensions$ dimensions and bound $\lambda\in\IN$ 
can be solved in time
\[\bigo\left(\left(|L|\cdot (\maxpriority^\numdimensions + 1) \cdot
  (|C| + \numdimensions) ! \cdot 2^{|C| + \numdimensions}\cdot
\prod_{x\in C}(2c_x+1)\cdot (2\lambda + 1)^\numdimensions\right)^4\right)\]
with the approach of~\cite{MRS21}, where $c_x$ denotes the largest bound
to which clock $x\in C$ is compared in clock constraints of $\automaton$.
It follows that the algorithm, for a fixed number of dimensions, is polynomial
in the bound on the size of windows (i.e., exponential in the size of its
encoding). When solving TGs with bounded timed window objectives, the
complexity of the algorithms presented in previous sections is not affected
by the potential size of good windows. Because the winning set for
a bounded objective coincides with the winning set for some fixed objective,
it follows that the algorithms for TGs with bounded objectives can be used
to more efficiently solve TGs with fixed objectives with large bounds, by
entirely bypassing the bound in question.

\bibliography{bib}

\newpage 
\appendix

\section{Winning strategies in \texorpdfstring{$\omega$}{omega}-regular timed games}\label{appendix:strategies}

{

\newcommand{\expandedstates}{\widehat{S}}
\newcommand{\expandedjd}{\widehat{\jointdestination}}
\newcommand{\expandeddpastates}{\widehat{\dpastates}}
\newcommand{\expandeddpatransitions}{\widehat{\mathsf{up}}}
\newcommand{\blame}{\mathsf{blame}}

In this section, we present an approach to solving timed games with
$\omega$-regular region objectives as a direct extension of the technique
of~\cite{AlfaroFHMS03} for timed games with $\omega$-regular location
objectives, i.e., objectives the satisfaction of which depends only on the
sequence of witnessed locations in the same way that region objectives depend
only on the sequence of witnessed regions along a play.
The main interest of this presentation is to highlight some
useful properties of winning strategies in timed games with $\omega$-regular
region objectives, namely that finite-memory region strategies suffice
for winning.
We assume that the objectives are given by deterministic parity automata.

The main ideas are as follows. First, we consider an expanded game in which
blamelessness and time-divergence can be encoded as $\omega$-regular conditions.
We alter the deterministic parity automaton defining the objective
so that it encodes the winning condition itself rather than the objective.
We can then compute memoryless
region strategies on the (infinite) parity game obtained through the
product of the expanded game and expanded parity automaton~\cite{AlfaroFHMS03,DBLP:conf/concur/AlfaroHM01}.
The remainder of this section
is devoted to showing that we can use these memoryless region strategies to
derive a winning finite-memory strategies on the non-expanded TG.

We fix a TG $\game = (\automaton, \Sigma_1, \Sigma_2)$ where
$\automaton = (L, \ell_\init, C, \Sigma_1\cup\Sigma_2, I, E)$
for this entire section. Recall that we use $S$ and $\to$ to denote the state
space and transition relation of $\mathcal{T}(\automaton)$, and
$\jointdestination$ for the joint-destination function.

\subparagraph*{Expanding the state space of the game.} To encode time-divergence
and blamelessness as $\omega$-regular conditions, we expand the state space
$S$ with two Boolean values, i.e., we consider an expanded state space
$\expandedstates = S\times \{\true, \false\}^2$. Expanded states are of the form
$(s, \tick, \blame)$, where $\tick$ holds if and only if during the previous
transition, the global clock $\gamma$ passes a new integer bound, and
$\blame$ holds if $\player_1$ is responsible for the latest transition.
We extend the joint-destination function so that it handles the additional
information. We denote by
$\expandedjd\colon\expandedstates\times M_1\times M_2\to 2^{\expandedstates}$
the expanded joint-destination function, defined as follows. For any expanded
state $\widehat{s} = (s, \tick, \blame)\in\expandedstates$ and moves
$m^{(1)}=(\delay^{(1)}, a^{(1)})\in M_1(s)$ and
$m^{(2)}=(\delay^{(2)}, a^{(2)})\in M_2(s)$ enabled in $s$, we set
\[\jointdestination(s, m^{(1)}, m^{(2)}) = \begin{cases}
    \{(s', \tick(s, \delay^{(1)}), \true) \mid s \xrightarrow{m^{(1)}} s'\} & \text{if } \delay^{(1)} < \delay^{(2)} \\
    \{(s', \tick(s, \delay^{(2)}), \false)\mid s \xrightarrow{m^{(2)}} s'\} & \text{if } \delay^{(1)} > \delay^{(2)} \\
    \{(s', \tick(s, \delay^{(i)}), \mathsf{bl}_1(s, m^{(1)}, s')) \mid s \xrightarrow{m^{(i)}} s', i=1, 2\} & \text{if }\delay^{(1)} = \delay^{(2)},
  \end{cases}\]
where for any $s' = (\ell, v)\in S$ and $\delay\geq 0$,
$\tick(s', \delay)$ holds if and only if
$\lfloor v(\globalclock)\rfloor < \lfloor v(\globalclock) + \delay\rfloor$, and
$\mathsf{bl}_1(s, m^{(1)}, s')$ holds only if $s \xrightarrow{m^{(1)}} s'$
(i.e., if $\player_1$ is responsible for the transition).

We denote this expanded game by $\widehat{\game}$.
The notions of plays, histories, time-divergence, blame, strategies and
objectives are defined analogously in $\widehat{\game}$ as they were in
regular TGs. We extend state equivalence to the state space of
$\widehat{\game}$ by saying that any two states
$(\ell, v, \tick, \blame)$ and $(\ell', v', \tick', \blame')$ are
state-equivalent if $\ell = \ell'$, $v\clockequiv v$, $\tick = \tick'$ and
$\blame = \blame'$. In other words, a state region of this expanded state
space is a set of
the form $\{\ell\}\times  R \times \{\tick\}\times \{\blame\}$ where
$\ell \in L$, $R\in\regions$, $\tick$, $\blame\in\{\true, \false\}$, i.e.,
obtained by taking a state region and adding two fixed Boolean values for the
last components.

We can define time-divergence and blamelessness
as $\omega$-regular conditions using the two additional Boolean values.
A play of the expanded game
is time-divergent if and only if infinitely many states of the form
$(s, \true, \blame)$ appear along it (i.e., the global clock passes infinitely
many integer bounds along the play). A play is blameless if and only if from
some index on, only states of the form $(s, \tick, \false)$ are visited, i.e.,
if from some point on, $\player_1$ is no longer responsible for transitions.

\subparagraph{Parity automata for winning conditions.}
We consider $\omega$-regular objectives specified by deterministic parity
automata. We explain how to derive
a DPA encoding the winning condition using the
additional information of $\widehat{\game}$
 from a DPA specifying a region objective
in the TG $\game$.

Let us fix a DPA
$H = (\dpastates, q_\init, L\times \regions, \dpatransitions, p)$.
One can derive from $H$ a DPA $\widehat{H}$
encoding the winning condition $\wc_1(\mathcal{L}(H))$
in the expanded game $\widehat{\game}$. Formally, we define
$\widehat{H} = (\expandeddpastates, \widehat{q}_\init, (L\times \regions) \times \{\true, \false\}^2, \expandeddpatransitions, \widehat{p})$, where
$\expandeddpastates = \dpastates\times \{\true, \false\}^2\times\{0, \ldots, \maxpriority-1\}$, $\widehat{q}_\init = (q_\init, \false, \false, d)$,
for any $\widehat{q} = (q, \tick, \blame, h)\in \expandeddpastates$ and
$\widehat{s} = ([s], \tick', \blame')$, we have
\[\expandeddpatransitions(\widehat{q}, \widehat{s}) =
  \begin{cases}
    (q', \tick', \blame', p(q')) & \text{if $\tick = \true$} \\
    (q', \tick', \blame', \min\{h, p(q')\}) & \text{otherwise},
  \end{cases}\]
where $q' = \dpatransitions(q, [s])$, and
\[\widehat{p}(\widehat{q}, \widehat{s}) =
  \begin{cases}
    h & \text{if $\tick = \true$} \\
    \maxpriority' & \text{if $\tick = \false$ and $\blame = \true$}, \\
    \maxpriority' + 1 & \text{otherwise}
  \end{cases}\]
where $\maxpriority' = \maxpriority$ if $\maxpriority$ is odd, and
otherwise $\maxpriority' = \maxpriority - 1$.
The DPA $\widehat{H}$ encodes an objective of $\widehat{\game}$ in the same
way that $H$ encodes an objective of $\game$. This objective is
the winning condition for the following reasons.

The rough idea of the construction is to keep track
of the smallest priority in $H$ seen between two ticks and output it whenever
$\tick$ holds. This way, whenever $\tick$ holds infinitely often, the smallest
priority appearing in an execution of $\widehat{H}$ is the same as the
smallest priority in the matching execution of $H$, because we chose $\maxpriority'$
greater or equal to all of the priorities of $H$.

If $\tick$ holds finitely
often however (i.e., we consider a time-convergent play), from some point on
only the priorities $\maxpriority'$ and $\maxpriority'+1$ are seen. We see the smaller odd
priority $\maxpriority'$ whenever $\player_1$ is responsible for a transition; it follows
that, in this case,  we have a rejecting execution of $\widehat{H}$
if and only if $\player_1$ is not blameless.

For the remainder of this section, we fix a DPA
$H= (\dpastates, q_\init, L\times \regions, \dpatransitions, p)$
and let $\widehat{H}$ denote its adaptation as defined above.

\subparagraph*{Computing the set of winning states.} We explain how to compute
the set of winning states of $\widehat{\game}$. The idea is to solve an
infinite parity game obtained via the synchronous product of the expanded game
$\widehat{\game}$ with the expanded DPA $\widehat{H}$. This approach is
presented in~\cite{DBLP:conf/concur/AlfaroHM01} and underlies the algorithmic
solution of~\cite{AlfaroFHMS03}.

The synchronous product of $\widehat{\game}$ and $\widehat{H}$, which we will
denote by $\widehat{\game}\times \widehat{H}$, is obtained  in the usual way.
At each step of the TG $\widehat{\game}$, we feed the state region, tick
and blame components of the current state to the DPA $\widehat{H}$. In the
sequel, because the tick and blame components in both $\widehat{\game}$ and
$\widehat{H}$ coincide (by nature of the product), we omit one of the two
in the upcoming definitions.

Formally, we obtain a game played on the state space
$S \times \widehat{Q}$, with the joint destination function
$\widehat{\jointdestination}_{\times}\colon S\times\widehat{Q}\times M_1\times M_2\to 2^{S\times\widehat{Q}}$ defined by, for all
$(s, \widehat{q})\in S\times\widehat{Q}$,
$\widehat{q} = (q, \tick, \blame, h)$, and all $m^{(1)}\in M_1(s)$ and
$m^{(2)}\in M_2(s)$,
\[\widehat{\jointdestination}_\times((s, \widehat{q}), m^{(1)}, m^{(2)}) =
  \{(s', \widehat{q}')\mid
  \widehat{s}' = (s', \tick', \blame')\in
  \widehat{\jointdestination}(\widehat{s}, m_1, m_2) \land
  \widehat{q}' = \widehat{\dpatransitions}(\widehat{q}, [\widehat{s}'])\},\]
where $\widehat{s} = (s, \tick, \blame)$.

On this product game, the objective of $\player_1$ is a parity objective.
The priority function $\widehat{p}_\times$ from which this objective is defined
assigns to each state $(\widehat{s}, q, h)$ the priority
$\widehat{p}(\widehat{q})$.

Winning in the product game $\widehat{\game}\times\widehat{H}$
and winning in the original TG $\game$ are related as follows.
There is a winning strategy for $\player_1$ in a state $s\in S$
in $\game$ for the (winning condition induced by the) objective encoded by $H$
if and only if there is a winning strategy for $\player_1$ from the
state $(s, \dpatransitions(q_\init, [s]), \false, \false, d-1)$ in
$\widehat{\game}\times\widehat{H}$ for the parity
objective given by $\widehat{p}_{\times}$. This can be
established by showing that from any winning strategy in the product game, one can
derive a winning strategy in the original TG and vice-versa.

The set of winning states in the product game can be computed by a linear-size
$\mu$-calculus formula of alternation depth
$\maxpriority'+2\leq \maxpriority+2$~\cite{DBLP:conf/concur/AlfaroHM01}.
Furthermore, and all sets involved in its computation are unions of state
regions~\cite{AlfaroFHMS03}, i.e., its evaluation can be performed on the
finite region abstraction (albeit of the product game). The following result
follows.
\theoremOmegaRegularComplexity*

Let us now discuss winning strategies in the product game. A strategy is said to
be \textit{memoryless} if for any two histories ending in the same state,
the same move is prescribed. In parity games, memoryless strategies suffice
for winning (e.g.,~\cite{DBLP:conf/focs/EmersonJ91,Zielonka98}).
In the product game $\widehat{\game}\times\widehat{H}$, one can find
winning memoryless strategies that are well-behaved with respect to regions.
A memoryless strategy $\sigma\colon S\times\widehat{Q}\to M_1$ is said to be
a \textit{memoryless region strategy} if for any two states
$(s_1, \widehat{q})$, $(s_2, \widehat{q})\in S\times\widehat{Q}$, where
$s_1 = (\ell_1, v_1)$ and $s_2 = (\ell_2, v_2)$, if $s_1\clockequiv s_2$,
then the moves
$(\delay_1, a_1) = \sigma((s_1, \widehat{q}))$ and
$(\delay_2, a_2) = \sigma((s_2, \widehat{q}))$ are such that $a_1 = a_2$,
$[v_1 + \delay_1] = [v_2 + \delay_2]$ and
$\{[v_1 + \delay_{\mathsf{mid}}]\mid 0\leq \delay_{\mathsf{mid}}\leq \delay_1\} =
\{[v_2 + \delay_{\mathsf{mid}}]\mid 0\leq \delay_{\mathsf{mid}}\leq \delay_2\}$.
Such memoryless region strategies suffice for winning
in $\widehat{\game}\times\widehat{H}$.

A function $f\colon (L\times\regions)\times\widehat{Q}\to U$, where $U$ denotes
the set of unions of elements of $(L\times\regions)\times\widehat{Q}$,
can be derived during the evaluation of the $\mu$-calculus
formula mentioned above.
This function $f$ describes a memoryless winning strategy at the region
level~\cite{DBLP:conf/concur/AlfaroHM01,AlfaroFHMS03}; a memoryless
winning strategy is obtained by assigning to any winning state
$(s, \widehat{q})\in S\times\widehat{Q}$ some move $m^{(1)}$ such that for any
move $m^{(2)}$ of $\player_2$ enabled in $(s, \widehat{q})$, we have
$\jointdestination_\times((s, \widehat{q}), m^{(1)}, m^{(2)})\subseteq f(([s], \widehat{q}))$ -- such a move is guaranteed to exist assuming that
$\player_1$ has a winning strategy from $(s, \widehat{q})$.

We explain how a memoryless winning region strategy can be obtained from $f$.
The choice of moves only matters in regions from which $\player_1$ wins.
Fix a state $(s, \widehat{q})\in S\times\widehat{Q}$ with $s = (\ell, v)$ and
let $m = (\delay, a)$ be any move that could have been assigned in
$(s, \widehat{q})$ by a winning strategy derived from $f$. Let
$s' = (\ell', v')\in S$ such that $s'\clockequiv s$; we argue that we can find
a move $m'=(\delay', a)$
such that
$[v + \delay] = [v' + \delay']$,
$\{[v + \delay_{\mathsf{mid}}]\mid 0\leq \delay_{\mathsf{mid}}\leq \delay\} =
\{[v' + \delay_{\mathsf{mid}}]\mid 0\leq \delay_{\mathsf{mid}}\leq \delay'\}$ and
for any move $m^{(2)}$ of $\player_2$ enabled in $(s', \widehat{q})$, we have
$\jointdestination_\times((s', \widehat{q}), m', m^{(2)})\subseteq f(([s'], \widehat{q}))$.
The properties of clock regions ensures that there exists some $\delay'$
satisfying the first two conditions. Fix any such $\delay'$. The
third condition follows from the facts that (i) $s\clockequiv s'$
implies $f(([s], \widehat{q})) = f(([s'], \widehat{q}))$ and (ii)
$m$ and $m'$ traverse and reach the same regions, therefore if $\player_2$ has
a move $(\delay_2', b)$ enabled in $s'$ with $\delay'_2\leq\delay'$, then
there is some $\delay_2\leq \delay$ such that
$[v + \delay_2] = [v' + \delay_2']$, therefore the sets of regions
$\{[s'']\mid (s'', \widehat{q}')\in
\jointdestination_\times((s, \widehat{q}), m, (\delay_2, b))\}$ and
$\{[s'']\mid (s'', \widehat{q}')\in
\jointdestination_\times((s', \widehat{q}), m', (\delay_2', b))\}$ are the same,
which implies the third condition in conjunction with (i).

\subparagraph*{Simplifying the structure of winning strategies.} In the previous
section, we have explained that in the product parity game
$\widehat{\game}\times \widehat{H}$, memoryless region strategies
suffice and are computable. To replicate the behavior of these strategies
in the original TG $\game$, one needs to observe the moves of the players,
e.g., to keep track of the blame component. The goal of this section is to show
that we can simplify winning strategies in two regards, with the goal of
deriving finite-memory strategies that do not take in account the moves of
the players.

Let us fix a memoryless winning region strategy
$\sigma\colon S\times \widehat{Q}\to M_1$ of $\player_1$ in the product game.
First, we show that the blame component is irrelevant to the decision of the
strategy. Formally, we show that we can select a winning strategy such that
if two expanded states $(s, \widehat{q}_1)$ and $(s, \widehat{q}_1)$ differ
only in their blame Boolean value, then the strategy prescribes the same move in
both states. Second, we show that we can bound the delays proposed by a winning
strategy, in such a way that whether $\tick$ holds or not can be inferred
without examining the delays in the moves.

To show that the blame Boolean can be disregarded, we provide a non-constructive
argument. The essence of the argument is that one can find a winning strategy
which assigns the same move to two states that possess the same successors.

\begin{lemma}\label{lemma:strategies:blame}
  In the game $\widehat{\game}\times \widehat{H}$, region memoryless strategies
  that disregard the blame Boolean suffice for winning.
\end{lemma}
\begin{proof}
  In (potentially infinite) turn-based parity games, memoryless strategies
  that select the same
  action in two states with the same successors suffice for winning; this
  follows from the proof of Emerson and
  Jutla~\cite{DBLP:conf/focs/EmersonJ91} that memoryless strategies suffice
  in turn-based parity games with finitely many priorities.

  While the game $\widehat{\game}\times \widehat{H}$ is not turn-based, the
  definition of winning we use (i.e., winning no matter the strategy of
  $\player_2$) allows us to apply the previous result. Indeed, winning in
  the concurrent product game $\widehat{\game}\times \widehat{H}$ is trivially
  equivalent to
  winning in a turn-based game in which first $\player_1$ selects a move,
  and then $\player_2$ is informed of $\player_1$'s move and has the choice
  to preempt $\player_1$ or to let $\player_1$'s move induce the next
  transition.
  
  In light of the above and the fact that two states that differ only from
  their blame component possess the same successors, this ends the proof.
\end{proof}

While the argument above is non-constructive, memoryless winning region
strategies are constructed in practice using algorithms for finite parity
games~\cite{DBLP:conf/concur/AlfaroHM01}.
One can show that winning strategies constructed by Zielonka's recursive
algorithm~\cite{Zielonka98} can be built such that two successor-sharing
states are assigned the same action. This is due to the fact that the building
blocks of these winning strategies are so-called attractor strategies, and
intuitively that successor-sharing states are in the same attractor sets
when neither are target states.

We now move on to the second step in our simplification of winning strategies.
The goal of the upcoming construction is to have ticks be detectable by
observing only the current state region and using one bit of information.
The role of the bit of information is to remember whether the valuation of
the global clock was integral or not at the previous step. 
This allows us to infer that $\tick$ holds in some cases:
$\tick$ holds whenever the valuation of the global clock is integral
at the current step but was not at the previous step.

In the sequel, we show that the delays proposed by a strategy can be
constrained in such a way that all ticks are detectable by the mechanism
described above.
Intuitively, our construction consists, given a memoryless winning strategy
that disregards the blame Boolean, to replace proposed moves that have a
large delay by delay moves with small delays in such a way that the strategy
obtained this way is still winning, and that all ticks are observable. 

It remains to clarify what we mean by a large delay. On the one hand,
any delay such that the global clock passes an integer bound strictly
is considered large; we cannot observe that the global clock
was integral at some point in time during the transition in this case.
On the other hand, a delay of one is also considered large:
from regions, we can only observe whether the valuation of the global clock is
integral or not. If we move between two states in which the valuation
of the global clock is integral, it cannot be known without observing the
moves whether the transition was taken with a non-zero delay or not,
therefore ticks cannot be observed.

We formally state and prove the announced result hereunder. Let us
underline that in the following proof, to lighten
notation, we denote by $\widehat{s}$ states of the product game
$\widehat{\game}\times \widehat{H}$. In previous sections, we had used
such a notation for states of the expanded game $\widehat{\game}$.
\begin{lemma}\label{lemma:strategies:delay}
  In the game $\widehat{\game}\times \widehat{H}$, region memoryless strategies
  $\sigma$ that satisfy the following constraints suffice for winning:
  $\sigma$ disregards the blame Boolean and for any state
  $\widehat{s} =
  ((\ell, v), \widehat{q})\in
  S\times \widehat{Q}$,
  we have $\delayfunc(\sigma(\widehat{s})) \leq 1 - \fracpart(v(\globalclock))$,
  and this inequality is strict whenever $v(\globalclock)\in\IN$.
\end{lemma}
\begin{proof}
  Let $\widehat{s}_\init\in  S\times \widehat{Q}$ be a state from which
  $\player_1$ wins.
  Let $\sigma$ denote a memoryless region strategy winning from
  $\widehat{s}_\init$ that disregards the blame
  Boolean, the existence of which is ensured by
  Lemma~\ref{lemma:strategies:blame}. We explicitly derive a suitable strategy
  $\widetilde{\sigma}$ from $\sigma$ and show it is winning.

  For any state $\widehat{s} = ((\ell, v), \widehat{q})\in S\times \widehat{Q}$,
  we let $f = \fracpart(v(\globalclock))$ and define
  \[\widetilde{\sigma}(\widehat{s}) = \begin{cases}
      \sigma(\widehat{s}) & \text{if $\delayfunc(\sigma(\widehat{s})) \leq  1 - f$ and $v(\globalclock)\notin\IN$} \\
      (1 - f, \bot) & \text{if $\delayfunc(\sigma(\widehat{s})) >  1 - f$ and $v(\globalclock)\notin\IN$} \\
      \sigma(\widehat{s}) & \text{if $\delayfunc(\sigma(\widehat{s})) <  1 - f$ and $v(\globalclock)\in\IN$} \\
      (\frac{1}{2}(1 - \max_{x\in C}\fracpart(v(x))), \bot) & \text{if $\delayfunc(\sigma(\widehat{s})) \geq   1 - f$ and $v(\globalclock)\in\IN$}.
    \end{cases}\]
  This memoryless strategy $\widetilde{\sigma}$
  disregards the blame Boolean because $\sigma$ does, and satisfies the
  delay-related constraints by construction. To end the proof, we must show
  that $\widetilde{\sigma}$ is a region strategy and that it is winning.

  First, let us show that it is a memoryless region strategy.
  Let $\widehat{s}_1$ and $\widehat{s}_2$ be two region-equivalent states.
  Because $\sigma$ is a region strategy, it proposes moves in both
  $\widehat{s}_1$ and $\widehat{s}_2$ that traverse and reach
  the same region. In particular, given that ticks are encoded in states in
  the product game $\widehat{\game}\times\widehat{H}$, and that cases in the
  definition of $\widetilde{\sigma}$ depend on whether $\tick$ holds or not
  after the move proposed by $\sigma$, it follows that both
  $\widehat{s}_1$ and $\widehat{s}_2$ fall into the same case.
  
  In the first or third cases, $\widetilde{\sigma}$ proposes the same move
  as $\sigma$, therefore there is nothing to show. We restrict for the
  remainder of this paragraph our attention to the set of clocks containing
  the global clock and the clocks such that
  their valuation in $\widehat{s}_1$ (or equivalently in $\widehat{s}_2$)
  has not yet exceeded the largest constant to which they
  are compared to in $\automaton$.
  Clocks for which the valuation has exceeded this
  threshold need not be taken in account to prove that the delays
  $\widetilde{\sigma}$ proposes traverse and reach the same regions from both
  $\widehat{s}_1$ and $\widehat{s}_2$ (by definition of regions).

  In the second and fourth cases, $\widetilde{\sigma}$ prescribes a delay
  move; it does not affect the ordering of the fractional parts of the
  valuations of the clocks. 
  It follows that we need only check that the same clocks have pass
  and reach an integral value during and after the delay prescribed by
  $\widetilde{\sigma}$ in $\widehat{s}_1$ and $\widehat{s}_2$ respectively.
  In the second case of the definition of $\widetilde{\sigma}$, the
  only clocks that have an integral valuation after
  the delay are those with the same fractional
  part in their valuation as $\globalclock$ by choice of the delay. Furthermore,
  the valuation of any clock that had a fractional part greater than that of
  $\globalclock$ before the delay passes an integer bound during the delay.
  In the fourth case, the chosen delay is such that the valuation of
  no clock passes an integer bound after the delay. This shows that in both
  cases, the same regions are traversed and reached from both states.
  This concludes the proof that
  $\widetilde{\sigma}$ is a region strategy.

  It remains to show that $\widetilde{\sigma}$ is winning to end the proof.
  The idea for the remainder of this proof is to show that for any outcome
  $\widetilde{\pi}$ of $\widetilde{\sigma}$ from $\widehat{s}_\init$,
  one can find an analogous outcome $\pi$ from $\widehat{s}_\init$
  of $\sigma$ (by changing the moves of $\player_2$)
  and use the fact that $\pi$ is winning to show that
  $\widetilde{\pi}$ is also winning.

  Let $\widetilde{\pi} =
  \widetilde{s}_0(\widetilde{m}_0^{(1)}, \widetilde{m}_0^{(2)})\widetilde{s}_2
  \ldots$ be an outcome of $\widetilde{\sigma}$. We consider the outcome
  $\pi = \widehat{s}_0(m_0^{(1)}, m_0^{(2)})\widehat{s}_2\ldots$ of
  $\sigma$ where $\widehat{s}_0 = \widetilde{s}_0$, and for all $k\in\IN$,
  $m_k^{(1)} = \sigma(\widehat{s}_k)$ and,
  if $m_k^{(1)} = \widetilde{m}_k^{(1)}$ or $\player_2$ is
  responsible for the transition at step $k$ in $\widetilde{\pi}$, we let
  $m_k^{(2)} = \widetilde{m}_k^{(2)}$ and
  $\widehat{s}_{k+1} = \widetilde{s}_{k+1}$, and otherwise, we let
  $m_k^{(2)} = \widetilde{m}_k^{(1)}$ (i.e., $\player_2$ takes over the
  delay move of $\player_1$) and $\widehat{s}_{k+1}$ is obtained by
  reversing the blame Boolean of $\widetilde{s}_{k+1}$.
  We note that the play $\pi$ is a well-defined play because by construction,
  $\widetilde{\sigma}$ proposes shorter delays than $\sigma$.
  Since we assume that $\sigma$ is winning, it follows that $\pi$ satisfies
  the parity objective.

  It now remains to show that $\widetilde{\pi}$ is winning for the parity
  objective. Assume first that $\tick$ holds infinitely often in
  $\widetilde{\pi}$. It follows
  by construction that $\tick$ holds infinitely often in $\pi$. In this
  case, the structure of the priority function of the product game
  $\widehat{\game}\times\widehat{H}$ ensures that the smallest priority
  occurring infinitely often in $\widetilde{\pi}$ and $\pi$ coincide, i.e.,
  $\widetilde{\pi}$ is winning for the parity objective.

  Let us now assume that $\tick$ holds only finitely often in $\widetilde{\pi}$
  and therefore in $\pi$.
  Because $\pi$ is winning, it follows that there exists an index
  $n\in\IN$ such that for all $k\geq n$, both the $\tick$ and $\blame$
  components of $\widehat{s}_k$ evaluate to false. By construction of
  $\pi$, for all $k\geq n$, the $\tick$ component of $\widetilde{s}_k$
  evaluates to $\false$. In the remainder of the proof, we argue that there
  is at most one $k\geq n$,
  such that the $\blame$ component of $\widetilde{s}_k$ is $\true$.

  Let us fix $k\geq n$.
  If $\widetilde{m}_{k-1}^{(1)} = m_{k-1}^{(1)}$,
  we have $\widetilde{s}_k = \widehat{s}_k$, hence the $\blame$ component
  of $\widetilde{s}_k$ evaluates to $\false$. Let us assume instead that
  $\widetilde{m}_{k-1}^{(1)} \neq m_{k-1}^{(1)}$. There are two possibilities:
  either the valuation of $\globalclock$ in $\widetilde{s}_{k-1}$ is not
  an integer or it is an integer. The former case is easiest to handle: by
  definition of $\widetilde{\sigma}$, because the move is changed, it means that
  the move $m_{k-1}^{(1)}$ has a delay large enough that $\tick$ would hold
  after using it, therefore the move $\widetilde{m}_{k-1}^{(1)}$ is defined
  in such a way that $\tick$ would hold after using it. However, because $\tick$
  does not hold in $\widetilde{s}_k$, it follows that $\blame$ does not hold
  either. Now, let us place ourselves in the latter case and assume that
  the valuation of $\globalclock$ is integral in $\widetilde{s}_{k-1}$. In this
  case $\widetilde{\sigma}$ proposes a delay move with a strictly positive
  delay. In these circumstances, it may be the case that $\player_1$ is
  responsible for the transition, but this can happen at most once: after
  one such transition, the valuation of the global clock is never an
  integer again as there are no more ticks.
  This shows that $\widetilde{\pi}$ is winning in this
  second case. This concludes the proof that the strategy $\widetilde{\sigma}$
  is winning, and with this, the entire proof.
\end{proof}

\subparagraph{Finite-memory strategies.} Up to now, we have been concerned
with winning strategies in the expanded game structure. In this section,
we describe how to derive winning finite-memory region strategies from the
memoryless winning region strategies on the product
$\widehat{\game}\times\widehat{H}$.
The role of the Mealy machine is to keep track of the additional information
contained in the expanded product game.

It follows from Lemma~\ref{lemma:strategies:delay} that to win in the expanded
product game, one can disregard the blame Boolean and restrict themselves to
delays that prevent the occurrence of two ticks in a row. Let us fix one
such winning memoryless region strategy
$\sigma\colon S\times \widehat{Q}\to M_1$ for the remainder of this section.

The structure of the Mealy machine encoding the finite-memory strategy
we derive from
$\sigma$ is very close in nature to the structure of $\widehat{H}$. The main
difference is that neither ticks nor blame are observable from state regions
in the TG $\game$, which is why we simplified strategies
to overcome these limitations.

We provide the construction of the Mealy machine in the proof of the
following formal statement.

\theoremFMStrategies* 

\begin{proof}
{
  \newcommand{\integer}{\mathsf{int}}
    
  It suffices to show that using a finite-memory strategy, it is possible to emulate
  $\sigma$ in $\game$. Any strategy constructed this way is winning due to the
  relations between the games $\game$ and $\widehat{\game}\times\widehat{H}$.
  The state space of the upcoming Mealy machine is a essentially
  a simplification of $\widehat{Q}$: states are of the form
  $(\integer, q, h)$ where $\integer\in\{\true, \false\}$ holds if and only
  if the valuation
  of the global clock was integral at the previous step, $q\in Q$ is some
  state of $H$ and $h\in\{0, \ldots, \maxpriority - 1\}$ is the smallest
  priority seen
  since the last tick. We also use $h= \maxpriority - 1$ in case no priorities
  were seen, i.e., if a tick has occurred in the current step.

  We formally define the Mealy machine
  $\fmmealy=(\fmstates, \mathfrak{m}_\init, \update, \nextmove)$
  as follows. The state space is
  $\fmstates = \{\true, \false\}\times \dpastates\times \{0, \ldots, \maxpriority - 1\}$
  and the initial state is
  $\mathfrak{m}_\init = (\true, q_\init, p(q_\init))$.
  Prior to defining the update function
  $\update\colon \fmstates\times (L\times \regions)\to\fmstates$
  and next-move function
  $\nextmove\colon \fmstates\times S\to M_1$,
  we introduce some notation.  
  Let $\mathfrak{m} = (\integer, q, h)\in \fmstates$ and
  $s = (\ell, v) \in S$. We denote by $q' = \dpatransitions(q, [(\ell, v)])$
  the successor of $q$ in $H$ after reading $[s]$. We also
  let $\integer'$ hold if and only if $v(\globalclock)\in\IN$.
  The definition of $\update(\mathfrak{m}, [s])$ is:
  \[\update(\mathfrak{m}, [s]) = \begin{cases}
      (\integer', q' , \maxpriority - 1) &
      \text{if } \neg \integer \text{ and } \integer'\text{ hold}\\
      (\integer', q' , \min\{p(q'), h\}) &
      \text{otherwise}
    \end{cases}\]
  and the definition of $\nextmove(\mathfrak{m}, s)$ is:
    \[\nextmove(\mathfrak{m}, s) = \begin{cases}
      \sigma(s, q', \true, \false, \min\{h,p(q')\}) &
      \text{if } \neg \integer \text{ and } \integer'\text{ hold} \\
      \sigma(s, q', \false, \false, \min\{h, p(q')\}) &
      \text{otherwise.}
    \end{cases}\]

  The Mealy machine that $\fmstates$ encodes a finite-memory
  region strategy because $\sigma$ is a region strategy.
  We now briefly explain why $\fmmealy$ encodes a strategy in $\game$
  with the same behavior as $\sigma$. This implies that the encoded strategy
  is winning.

  In the product game $\widehat{\game}\times\widehat{H}$, updates
  of the DPA component are performed using the state we move into. Given that
  this is not possible in practice (we do not know in advance where we will
  be at the next step), the Mealy machine is always one step behind. This
  explains why in the evaluation of $\sigma$ used in the definition of 
  $\nextmove$, we use $q'$ rather than $q$ and the priority
  of $q'$ in the last argument of the function.
  By choice of $\sigma$, two ticks cannot occur consecutively, therefore
  at some step, $\tick$ holds if and only if $\integer$ did not hold previously
  and holds now. This justifies the two distinguished cases in the definitions
  of both $\update$ and $\nextmove$.

  Finally, it remains to explain that the priority fed to $\sigma$
  (i.e., last component in the evaluation of $\sigma$) is
  well-chosen. Whenever a tick is registered by $\widehat{H}$, the last
  component is reset to the priority of the successor state in the execution of
  $H$ following the current history.
  In our case, we cannot guess what will be this priority will be in advance.
  Instead, we set $\maxpriority - 1$ as the current lowest priority after a tick. This way,
  this greatest priority is disregarded by the $\min$ operator in the
  definition of $\nextmove$.

}
\end{proof}

\begin{remark}
  The assumption of a global clock $\globalclock$ is crucial for the
  existence of finite-memory winning strategies. In fact, if one removes
  this global clock, winning may require infinite memory and observing the
  moves~\cite{DBLP:conf/hybrid/ChatterjeeHP08}.
  Essentially, this infinite memory can be described as a simulation
  of the clock $\globalclock$.
\end{remark}

\end{document}

%% file: macros.tex
\newcommand{\IN}{\mathbb{N}}

\newcommand{\IR}{\mathbb{R}}

\DeclareMathOperator*{\argmin}{argmin}

\newcommand{\fracpart}{\mathsf{frac}}

\newcommand{\automaton}{\mathcal{A}}
\newcommand{\game}{\mathcal{G}}
\newcommand{\player}{\mathcal{P}}

\newcommand{\plays}{\mathsf{Plays}}
\newcommand{\hist}{\mathsf{Hist}}
\newcommand{\last}{\mathsf{last}}
\newcommand{\outcome}{\mathsf{Outcome}}

\newcommand{\dpastates}{Q}
\newcommand{\dpatransitions}{\mathsf{up}}
\newcommand{\fmmealy}{\mathcal{M}}
\newcommand{\fmstates}{\mathfrak{M}}
\newcommand{\nextmove}{\alpha_\mathsf{mov}}
\newcommand{\update}{\alpha_\mathsf{up}}

\newcommand{\reset}{\mathsf{reset}}
\newcommand{\delayfunc}{\mathsf{delay}}
\newcommand{\action}{\mathsf{action}}
\newcommand{\regions}{\mathsf{Reg}}
\newcommand{\tick}{\mathsf{tick}}
\newcommand{\globalclock}{\gamma}
\newcommand{\clockequiv}{\mathrel{\equiv_\automaton}}

\newcommand{\wc}{\mathsf{WC}}

\newcommand{\true}{\mathsf{true}}
\newcommand{\false}{\mathsf{false}}

\newcommand{\bigo}{\mathcal{O}}

\newcommand{\numdimensions}{\mathsf{K}}
\newcommand{\maxpriority}{\mathsf{D}}
\newcommand{\init}{\mathsf{init}}
\newcommand{\jointdestination}{\mathsf{JD}}
\newcommand{\delay}{\delta}